\documentclass[letterpaper]{article}
\usepackage{uai2020}
\usepackage[margin=1in]{geometry}

\usepackage{times}
\usepackage{graphicx}
\usepackage{booktabs}
\usepackage{amsmath}
\usepackage{mathtools}
\usepackage{amsfonts}
\usepackage{amssymb}
\usepackage{microtype}
\usepackage{bbold}
\usepackage{graphics}
\usepackage{graphicx}

\usepackage{amsthm}
\usepackage{blindtext}
\usepackage{amssymb}
\usepackage{natbib}
\usepackage{hyperref}
\usepackage{graphicx}
\usepackage{algorithm}
\usepackage{algorithmic}
\usepackage[usestackEOL]{stackengine}
\usepackage{subfig}
\newtheorem{theorem}{Theorem}
\newtheorem{lemma}[theorem]{Lemma}
\newtheorem{proposition}[theorem]{Proposition}

\title{Distortion estimates for approximate Bayesian inference}

\author{ {\bf Hanwen Xing} \\
Department of Statistics \\
University of Oxford, UK \\
\And
{\bf Geoff K. Nicholls}  \\
Department of Statistics \\
University of Oxford, UK\\
\And
{\bf Jeong Eun Lee}   \\
Department of Statistics \\
University of Auckland, New Zealand\\
}

\begin{document}

\maketitle

\begin{abstract}
Current literature on posterior approximation for Bayesian inference offers many alternative methods. Does our chosen approximation scheme work well on the observed data? The best existing generic diagnostic tools treating this kind of question by looking at performance averaged over data space, or otherwise lack diagnostic detail. However, if the approximation is bad for most data, but good at the observed data, then we may discard a useful approximation. We give graphical diagnostics for posterior approximation at the observed data. We estimate a ``distortion map'' that acts on univariate marginals of the approximate posterior to move them closer to the exact posterior, without recourse to the exact posterior.
\end{abstract}

\section{INTRODUCTION}

When we implement Bayesian inference for even moderately large datasets or complicated models some approximation is usually inescapable. Approximation schemes suitable for different Bayesian applications include Approximate Bayesian Computation \citep{pritchard1999population, beaumont2010approximate}, Variational Inference \citep{jordan1999introduction, hoffman2013stochastic}, loss-calibrated inference \citep{lacoste2011approximate, kusmierczyk2019variational} and synthetic likelihood \citep{wood2010statistical,price2018bayesian}. New applications suggest new approximation schemes. In this setting diagnostic tools are useful for assessing approximation quality. 



\cite{menendez14} give procedures for correcting approximation error in Bayesian credible sets. \cite{rodrigues2018recalibration} give a post-processing algorithm specifically for recalibrating ABC posteriors. The new generic diagnostic tools given in \cite{yao2018yes} and \cite{talts2018validating} focus on checking the average performance of an approximation scheme over data space $\mathcal{Y}$ and are related to \citet{prangle2014diagnostic}, which focuses on ABC posterior diagnostics. Their methods can be seen as an extension of  \citet{cook2006validation}, \citet{geweke2004getting} and \citet{monahan1992proper}, which were setup for checking MCMC software implementation. In contrast, we are interested in the the quality of approximation \emph{at the observed data $y_{obs}$}. If one posterior approximation scheme works poorly in some region of $\mathcal{Y}$ but works well at $y_{obs}$, we may reject a useful approximation using any diagnostic based on average performance. We may conversely accept a poor approximation. 

We give a generic diagnostic tool which checks the quality of a posterior approximation specifically \emph{at} $y_{obs}\in {\mathcal{Y}}$. We assume that 1) we can efficiently sample parameters $x\in {\mathcal{X}}$ from both the prior distribution $\pi(x)$ and the observation model $p(y|x)$ and 2) the approximation scheme we are testing is itself reasonably computationally efficient. We need this second assumption as we may need to call the approximation algorithm repeatedly. 

The posterior has a multivariate parameter. However, we run diagnostics on one or two parameters or scalar functions of the parameters at a time, so our notation in
Sections~\ref{sec:dmap} and \ref{sec:mv} takes $x\in \mathcal{R}$ and $\mathcal{R}^2$ respectively. Parameters are continuous but this is not essential. 

We introduce and estimate a family of ``distortion maps'' \[D_{y}: [0,1] \longrightarrow [0,1],\ y\in \mathcal Y\] which act on univariate marginals of the multivariate approximate posterior. 
The exact distortion map transports the approximate marginal posterior CDF $G_{y_{obs}}(x)$ onto the corresponding exact posterior CDF $F_{y_{obs}}(x)$. The distortion map is a function of the parameter $x$ defined at each $y\in \mathcal Y$ by the relation $F_y=D_y\circ G_y$ given in Eqn.~\ref{eq:dmap} below.
The distortion map $D_{y_{obs}}$ at the data contains easily-interpreted diagnostic information about the approximation error in the approximate marginal CDF $G_{y_{obs}}$. If the distortion map $D_{y_{obs}}$ differs substantially from the identity map, then the magnitude and location of any distortion is of interest. Our distortion map is an optimal transport \citep{el2012bayesian} constructed from a normalising flow.

A reliable estimate of $D_{y_{obs}}$ must be hard to achieve, as it maps to the exact posterior CDF $F_{y_{obs}}$. We estimate a map $\hat D_{y_{obs}}$ to a distribution $\hat F_{y_{obs}}=\hat D_{y_{obs}}\circ G_{y_{obs}}$ which is only {\it asymptotically closer in KL-divergence} to $F_{y_{obs}}$, not equal to it.
If $G_{y_{obs}}$ is far from $F_{y_{obs}}$ in KL divergence then it is easy to find a distribution $\hat F_{y_{obs}}$ which is closer to $F_{y_{obs}}$ than $G_{y_{obs}}$ was. It follows that if $\hat D_{y_{obs}}$ differs significantly from the identity map then the approximation defining $G_{y_obs}$ was poor. In this approach we get diagnostically useful estimates of the distortion map without sampling or otherwise constructing the exact posterior. 

The map $D_y,\ y\in \mathcal Y$ may be represented in several ways, with varying convenience depending on the setting. We can parameterise a transport map from the approximate density to the exact density, or a mapping between the CDF's, or a function of the approximate random variable itself. Since we are not interested in approximating the true posterior, but in checking an existing approximation for quality, we map CDF's, estimating the distortion of the CDF for each marginal of the joint posterior distribution. This has some benefits and some disadvantages. 

On the plus side, the mapping from the CDF of the approximate posterior to the CDF of the exact posterior is an invertible mapping between functions of domain and range $[0,1]$. This resembles a copula-like construction (see in particular Eqn.~\ref{eq:jointdistortion}) and doesn't change from one problem to another, making it easier to write generic code. There is also a simple simulation based fitting scheme, Algorithm~1, to estimate the map. On the downside, we restrict ourselves to diagnostics for low-dimensional marginal distributions. However, multivariate posterior distributions are in practise almost always summarised by point estimates, credible intervals and univariate marginal densities, and the best tools we have seen, \cite{prangle2014diagnostic} and \cite{talts2018validating}, also focus on univariate marginals. We extend our diagnostics to bivariate marginal distributions in Sec.~\ref{sec:mv} and give examples of estimated distortion surfaces in examples below. This works for higher dimensional distortion maps, but it is not clear how this would be useful for diagnostics and it is harder to do this well.


\section{DISTORTION MAP}
\label{sec:dmap}
Let $\pi(\cdot)$ be the prior distribution of a scalar parameter $x \in {\mathcal{X}} \subseteq \mathcal{R}$ and let $p(\cdot| x)$ be the likelihood function of generic data $y \in \mathcal{Y}$. Let $y_{obs}$ be the observed data value. Given generic data $y$, let $F_y(x)$ be the CDF of the exact posterior $\pi(x|y)\propto \pi(x)p(y|x)$. In practice these densities will the marginals of some multivariate parameter of interest. For $X\sim \pi(\cdot)$ and $Y|(X=x)\sim p(\cdot|x)$,
we have $X|(Y=y)\sim \pi(\cdot|y)$. We assume $X|(Y=y)$ is continuous, so that $F_y(x)$ is continuously differentiable and strictly increasing with $x$ at every $y\in \mathcal Y$. The case of $X$ discrete is a straightforward extension. Let $\tilde \pi(x|y)$ be a generic approximate posterior on ${\mathcal{X}}$ with CDF ${G}_y(x)$. 
We define a distortion map $D_y: [0,1] \longrightarrow [0,1]$ such that for each $x \in {\mathcal{X}}$ and each $y\in \mathcal{Y}$
\begin{equation}\label{eq:dmap}
D_y({G}_y(x)) = F_y(x).
\end{equation}
The distortion map $D_y$ is a strictly increasing function mapping the unit interval to itself and, as \cite{prangle2014diagnostic} point out, is itself the CDF of $Q={G}_y(X)$ when $X\sim F_y$. To see this observe that since $F_y(X)\sim U(0,1)$ we have $D_y(Q)\sim U(0,1)$ from Eqn.~\ref{eq:dmap}, and this is necessary and sufficient for
$Q\sim D_y$.

Denote by \[d_y(q) = \frac{d}{dq} D_y(q)\] the density associated with the CDF $D_y$ so that $Q\in [0,1]$ is random variable with probability
density $d_y(q)$ for $q\in [0,1]$. Since $\pi(x|y)=\frac{d}{dx}F_y(x)$, we have from Eqn.~\ref{eq:dmap},
\begin{equation}\label{eq:postmap}
\pi(x|y)=d_y(G_y(x))\tilde \pi(x|y),
\end{equation}
connecting the two posterior densities.

We seek an estimate, $\hat{D}_{y_{obs}}$, of the true distortion map  at the data, 
or equivalently an estimate, $\hat d_{y_{obs}}$, of its density. 
Other authors, focusing on constructing new posterior approximations, 
have considered related problems, either without the distortion-map representation, or in an ABC setting. However, since we seek a diagnostic map, not a new approximate posterior, it is not necessary to estimate $D_y$ exactly, but simply to find an approximate $\hat D_y$ that moves $G_y$ towards $F_y$ as measured by KL-divergence. The recalibrated CDF 
\begin{equation}\label{eq:fhat}
\hat F_{y_{obs}}(x)=\hat D_{y_{obs}}(G_{y_{obs}}(x))
\end{equation}
should be a better approximation (in KL-divergence) to $F_{y_{obs}}$ than $G_{y_{obs}}$ was {\it even if both are bad}. The same argument applies at the level of densities. From Equation~\ref{eq:dmap}, the recalibrated density
\[
\hat \pi(x|y)=\hat d_y(G_y(x))\tilde \pi(x|y)
\] 
must improve on $\tilde\pi(x|y)$. If our original 
approximation $\tilde \pi(x|y)$ is bad, then we should be able to improve it.

Working with the distortion map $D_y(q)$ is very convenient for building generic code: our diagnostic wrapper, Algorithm~1 below, is always based on a model for a density in $[0,1]$.
In practice users will have a multivariate approximation $\tilde \pi(x^{(1)},...,x^{(p)}|y_{obs})$ and get diagnostics by simulating or otherwise computing marginals $\tilde \pi(x^{(i)}|y_{obs})$. This distribution is computationally tractable, in contrast to $\pi(x^{(i)}|y_{obs})$.

\section{ESTIMATING A DISTORTION MAP}


We now explain how we approximate the distortion map without simulating the exact posterior. 
The distortion map $D_y$ we would like to approximate is a continuous distribution on $[0,1]$ so one approach is to sample it and use the samples to estimate $D_y$. The difficulty is that $D_y(x)$ is a function of $x$ which varies from one $y$-value to another. We can proceed as in Algorithm~1 below which we now outline. 

We start by explaining how to simulate $Q\sim D_y$. If we simulate the generative model, $\{x,y\}\sim \pi(x)p(y|x)$, then by Bayes rule $\{x,y\}\sim p(y)\pi(x|y)$ with $p(y)=E_X(p(y|X))$ the marginal likelihood, so a simulation from the generative model gives us a draw $X$ from the exact posterior at the random data $Y=y$. This observation is just the starting point for ABC. Now, from our discussion below Equation~\ref{eq:dmap}, if $Q=G_{y}(X)$ then the pair $\{Q,Y\}$ have a joint distribution with density $d_y(q)p(y)$ and conditional distribution $Q|(Y=y)\sim D_y$.
This is a recipe to simulate $\{q_i,y_i\}_{i=1}^N$ pairs which are realisations of $\{Q,Y\}$:
Simulate
$\{x_i,y_i\}_{i=1}^N$ with $x_i\sim \pi(\cdot)$ and $y_i\sim p(\cdot|x_i)$ 
and then set $q_i=G_{y_i}(x_i)$ (the subscript $i=1,...,N$ runs over samples,
not multivariate components). If $\tilde\pi(x|y)$ admits a closed form CDF $G_{y}(x)$ then $q_i$ can be evaluated directly. If $G_{y}(x)$ is not tractable (as in our examples below) then we estimate it using MCMC samples from the approximate posterior. We form the empirical CDF $\hat G_y(x)$ 
and set $q_i=\hat G_{y_i}(x_i)$. The samples $\{q_i,y_i\}_{i=1}^N$ are our ``data'' for learning about $D_y$.

We next define a semi-parametric model for $D_y(q)$ and a log-likelihood for our new ``data''.  
For $q \in [0,1]$ and $w\in \mathcal{R}^m$ let $\mathcal{D}_m=\{D_y(\cdot;w); w\in \mathcal{R}^m\}$ be a family of continuously differentiable strictly increasing CDFs parameterised by an $m$-component parameter $w$, and including the identity map, $D_y(q;w_I)=q$, for some $w_I\in \mathcal{R}^m$ and all $q\in [0,1]$.
Because we are parameterising the distortion, we are working with a probability distribution on $[0,1]$, so we simply model $d_y(q;w)$, the corresponding density of $D_y(q;w)$, using
\begin{equation}\label{eq:beta}
    d_y(q;w)=\mbox{Beta}(q;a(y;w),b(y;w)),
\end{equation}
a Beta density with parameters $a=a(y;w)$ and $b=b(y;w)$ which vary over $\mathcal Y$. The functions
$a,b:\mathcal{Y}\rightarrow (0,\infty)$ are parameterised by a feed-forward 
neural net with two hidden layers and positive outputs $a$ and $b$. We tried a Mixture Density Network (MDN) \citep{bishop1994mixture} of Beta-distributions but found no real gain from taking more than one mixture component. 

We now fit our model and estimate $D_{y_{obs}}$ at the data.
The log likelihood for our parameters given our model $D_y(q;w)$ and simulations $\{q_i,y_i\}_{i=1}^N$ is
\begin{equation}\label{eq:distortionlkd}
\ell(w; \{{q}_i,y_i\}_{i=1}^N) = \frac{1}{N}\sum_{i=1}^N \log {d}_{y_i}({q}_i;w).
\end{equation}

Let ${\hat w}_N$ maximise this log-likelihood and consider the estimate $\hat D_{y_{obs}}(q)=D_{y_{obs}}(q;{\hat w}_N),\ q\in [0,1]$. 
Let 
\begin{equation}\label{eq:W}
W=\{w^*\in \mathcal{R}^m: D_y(q)=D_y(q;w^*)\}
\end{equation} 
be the set of parameter values giving the true distortion map. This set is empty unless $D_y\in \mathcal D_m$, so the true map can be represented by the neural net.

We show below that, if the neural net is sufficiently expressive, so that $W$ is non-empty,
then $D_y(q;\hat{w}_N) \stackrel{p}{\rightarrow} D_y(q)$ for any fixed $\{y,q\}$. This is not straightforward as $w^*$ in Eqn.~\ref{eq:W} is in general not identifiable so standard regularity conditions for MLE-consistency are not satisfied. Our result compliments that of \cite{papamakarios2016fast} and \cite{2019arXiv190507488G}. Working in a similar setting, those authors show that the {\it maximiser of the limit} of the scaled log-likelihood gives the true distortion map (if the neural net is sufficiently expressive). Our consistency proof shows that the {\it limit of the maximiser} $\hat{w}_N$ converges to the set $W$ of parameter values that express the true distortion map. 

Proposition 1 translates the result of \cite{papamakarios2016fast} to our setting. 
At $y\in \mathcal{Y}$ and fixed $w\in \mathcal{R}^m$, the exact and approximate distortion maps, $D_y(q)$ and $D_y(q;w)$ have
associated densities $d_y(q)$ and $d_y(q;w)$. Their KL-divergence is
\[\mbox{KL}(D_y(\cdot),D_y(\cdot;w))\equiv \int_0^1 d_y(q)\log\left(\frac{{d}_y(q)}{{d}_y(q;w)}\right) dq.\]
Here, as in \cite{papamakarios2016fast}, the 
KL-divergence of interest is the complement of that used in variational inference.
We choose the approximating distribution $D_y(\cdot;w)$ to fit samples drawn from the true distribution $D_y(\cdot)$. This is possible using ABC-style joint sampling of $x$ and $y$. By contrast in variational inference $D_y(\cdot;w)$ is varied so that \emph{its} samples match $D_y(\cdot)$.
\begin{proposition} Suppose the set $W$ in Equation~\ref{eq:W} is non-empty. Let $y_i \sim p(y), \ q_i \sim D_{y_{i}}(q)$ independently for $i=1,...,N$. Then $N^{-1}\ell(w,\{q_i,y_i\}_{i=1}^N)$ converges in probability to
\[-E_Y(\mbox{KL}(D_Y(\cdot),D_Y(\cdot;w)))+E_{Q,Y}(\log(d_Y(Q)).\] 
This limit function is maximized at $w\in W$.
\end{proposition}
We can remove the condition that $W$ is non-empty in Proposition~1.
This leads to modified versions of the lemma and theorem below which may be more relevant in practice. This is discussed in Appendix. 

Proposition~1 tells us that we are maximising the right function, since the limiting KL divergence is minimised at the true distortion map $D_Y$, but it does not show consistency for $D_y(\cdot;\hat w_N)$. 
In Lemma 1 we prove that $D_y(q;\hat{w}_N)$ is a consistent estimate of $D_y(q)$.
\setcounter{theorem}{0}
\begin{lemma}
Under the conditions of Proposition~1, the estimate $D_y(q;{\hat w}_N)$ is consistent, that is
\[
\lim_{N\rightarrow\infty}\Pr(|D_y(q;{\hat w}_N)-D_y(q)|>\epsilon)=0.
\]
for every fixed $q,y$.
\end{lemma}

Our main result, Theorem 1, follows from Lemma~1. It states that, asymptotically, and in KL divergence, the ``improved" CDF $\hat F_y(x) = D_y(G_y(x);\hat{w}_N)$ is closer to the true posterior CDF $F_y(x)$ than the original approximation $G_y(x)$. All proofs are given in Appendix. 
\setcounter{theorem}{0}
\begin{theorem}
Under the conditions of Proposition~1 and assuming $KL(F_y,G_y)>0$, 
\[\Pr(KL(F_y,\hat F_y)<KL(F_y,G_y))\rightarrow 1\] as $N\rightarrow \infty$ for every fixed $y$.
\end{theorem}
The fitted distortion map at the data, $\hat D_{y_{obs}}(q)=D_{y_{obs}}(q;{\hat w}_N) = \mbox{Beta}(q;a(y;{\hat w}_N),b(y;{\hat w}_N))$ is of interest as a diagnostic tool.
The improved posterior CDF, $\hat F_y(x)$ in Equation~\ref{eq:fhat}, or the corresponding PDF $\hat{\pi}(x|y)$, is of only indirect interest to us. 
The point here is that $\hat D_{y_{obs}}$ may be a useful diagnostic for the approximate posterior even if $\hat F_y(x)$ is a poor approximation to $F_y$ as $\hat F_y(x)$ is at least asymptotically closer in KL-divergence to $F_y$ than $G_y$ is. 
If we can improve on the approximation $G_y$ substantially in KL-divergence to the true posterior, then it was not a good approximation. 

Plots of $d_{y_{obs}}(q;\hat w_N)$ give an easily interpreted visual check on the approximate posterior $\tilde \pi(x|y_{obs})$. A check of this sort is not a formal test, but such a test would not help as we \emph{know} $\tilde \pi(\cdot |y_{obs})$ is an approximation and want to know where it deviates and how badly. Since $D_y$ is a quantile map, if $d_{y_{obs}}(q;\hat w_N)$ is a cup shaped function of $q\in [0,1]$ then $G_y$ is under-dispersed, cap-shaped is over-dispersed, and if say $D_{y_{obs}}(1/2;\hat w_N)\gg 1/2$ then the median of $G_y$ lies above the median of $F_y$ and so this is evidence that $G_y$ is skewed to the right. 


When we apply Algorithm~1 we need good neural-net regression estimates $\hat D_y$ for $y$ in the neighborhood of $y_{obs}$ only. Fitting the neural net may be quite costly, and since the distortion-map estimate at $y_{obs}$ is in any case dominated by information from pairs $\{q,y\}$ at $y$-values close to $y_{obs}$, 
we regress on pairs $\{q_i,y_i\}$ such that $y_i \in \Delta$, where $\Delta\subseteq \mathcal{Y}$ is a neighbourhood of $y_{obs}$.
This is not ``an additional approximation'' and quite different to the windowing used in ABC. In our case our estimator is consistent for \emph{any} fixed neighborhood $\Delta$ of $y_{obs}$, whilst in ABC this is not the case. Extending the regression to the whole of $\mathcal{Y}$ space would be straightforward but pointless. 

\begin{algorithm}[t]
  \caption{Estimating the distortion map $D_{y_{obs}}$}
  \begin{algorithmic}
  \STATE \textbf{Input}: the observed data $y_{obs}$; functions evaluating summary statistics $s(y), y\in \mathcal{Y}$ and the approximate CDF ${G}_y(x)$; a subset $\Delta \subset \mathcal{Y}$ centered at $y_{obs}$; functions simulating the prior $\pi(x)$ and observation model $p(y|x)$.
   \FOR{$i$ in $1,\ldots,N$}
        \STATE sample $\{x_i,y_i\} \sim \pi(x)p(y|x)$ until $y_i \in \Delta$
        \STATE compute $q_i = {G}_{y_i}(x_i)$
   \ENDFOR
\STATE Fit a feed-forward net with weights $w\in\mathcal{R}^m$, 
input vector $s(y_i)\in\mathcal{R}^p$,
two scalar outputs $a(s(y_i);w), b(s(y_i);w)$ and loss function $-\ell(w; \{{q}_i,y_i\}_{i=1}^N)$ given by Eqns.~\ref{eq:beta} and \ref{eq:distortionlkd}.
\STATE \textbf{Return}:  the fitted distortion map 
$\hat D_{y_{obs}}(q)=D_{y_{obs}}(q;{\hat w}_N), q\in [0,1]$ where $\hat w_N$ are the fitted weights.
  \end{algorithmic}
  \footnotetext{Equivalent to sampling $\{x_i,y_i\} \sim \pi(x)p(y|x)\mathbb{1}(y\in \Delta)$, a truncated version of the generative model}

\end{algorithm}

Note that in Algorithm~1 we have introduced summary statistics $s(y)$ on the data.
This may be useful if the data are high dimensional, or where there is a sufficient statistic. In the examples which follow we found we were either able to train the network with $s(y)=y$, or had sufficient statistics in an exponential family model for a random network. 

\subsection{Validation checks on $\hat D_{y}$}\label{sec:checks}
In this section we discuss the choice of $N$ and the sample variation of $\hat D_y$. Since $\hat D_{y}(\cdot) = D_y(\cdot;\hat w_N)$ is consistent, $D_y(\cdot;\hat w_{N_j})$ converges in probability on any increasing subsequence $N_j, j=1,2,3...$. In order to check we have taken $N$ large enough so that taking it larger will not lead to significant change, we estimate $D_{y_{obs}}(\cdot;\hat w_{N_i})$ at a sampling of equally spaced $N_j$-values $N_0, N_1,...,N_J$ with $N_0=0$ and $N_J=N$ increasing up to $N$. We check that $D_{y_{obs}}(q;\hat w_{N_j})$ converges numerically at each $q\in [0,1]$ with increasing $j=1,...,J$ and is stable. In order to check the sample dependence, we break up our sample $\{q_i,y_i\}_{i=1}^N$ into blocks $\{q_i,y_i\}_{i=N_j+1}^{N_{j+1}}$ and, for $j=0,...,J-1$, form separate estimates $\hat D^{(j)}_{y_{obs}}$ and check the variation between function estimates is small.

\subsection{Extending to higher dimensions} \label{sec:mv}
In this section we show how to estimate distortion maps and the corresponding densities for the approximate posterior density $\tilde{\pi}(x_1,x_2|y)$ of a continuous bivariate parameter $(x_1,x_2) \in \mathcal{R}^2$. The extension to higher dimensions is straightforward but not obviously useful for diagnostics.

Let $G_{x_1,y}(x_2)$ and $F_{x_1,y}(x_2)$ be the CDF's of the approximate and exact conditional posteriors, respectively $\tilde \pi(x_2|x_1,y)$ and $\pi(x_2|x_1,y)$,
and let $G_{y}(x_1)$ and $F_{y}(x_1)$ be the CDF's of the approximate and exact marginal posteriors, respectively $\tilde \pi(x_1|y)$ and $\pi(x_1|y)$.
Let $D_{x_1,y}$ be the distortion map defined by
\begin{equation}\label{eq:conditionaldistortion}
    D_{x_1,y}(G_{x_1,y}(x_2))=F_{x_1,y}(x_2),
\end{equation}
with $D_y(G_y(x_1))=F_y(x_1)$ as before.
The transformation of the joint density is 
\begin{equation}\label{eq:jointdistortion}
    {\pi}(x_1,x_2|y) = d_{x_1,y}({G}_{x_1,y}(x_2)) d_y({G}_y(x_1))\tilde\pi(x_1,x_2|y)
\end{equation}
If the approximation is good at $y\in\mathcal Y$, then the densities $\pi(x_1,x_2|y)$ and $\tilde\pi(x_1,x_2|y)$ are near equal, which holds if the ``distortion surface'', $d_y(q_1,q_2)$ defined by 
\begin{equation}\label{eq:bivariatedistortion}
d_y(q_1,q_2)\equiv d_{G^{-1}_y(q_1),y}(q_2) d_y(q_1),
\end{equation}
is close to one for all arguments $(q_1,q_2)\in [0,1]^2$.

We estimate $D_y(q_1)$ as before. We estimate $D_{x_1,y}(q_2)$ by treating $x_1$ as data alongside $y$.
We apply Algorithm~1, but now we simulate $\{x_{1,i},x_{2,i},y_i\}$ from the generative model in the for-loop, and create two datasets. The first dataset, $\{q_{1,i},y_i\}_{i=1}^N$ with $q_{1,i}=G_{y_i}(x_{1,i})$, is the same as before. The second, $\{q_{2,i},(x_{1,i},y_i)\}_{i=1}^N$ with $q_{2,i}=G_{x_{1,i},y_i}(x_{2,i})$, is used to estimate the conditional $D_{x_1,y}$. We fit two neural network models for the Beta-density parameters, one fitting the Beta-CDF $D_y(q_1;w)$ using inputs $s(y_i)$ and choosing weights $w\in {\mathcal R}^{m_1}$ to maximise the likelihood
\begin{equation}\label{eq:conditionallkd}
\ell(w;\{q_{1,i},y_i\}_{i=1}^N)=\sum_{i=1}^N \log d_{y_i}(q_{1,i};w)
\end{equation}
and the other fitting the Beta-CDF $D_{x_1,y}(q_2;v)$ using inputs $(x_{1,i},s(y_i))$ and choosing weights $v\in {\mathcal R}^{m_2}$ to maximise the likelihood
\begin{equation}\label{eq:jointlkd}
\ell(v;\{q_{2,i},(x_{1,i},y_i)\}_{i=1}^N)=\sum_{i=1}^N \log d_{x_{1,i},y_i}(q_{2,i};v).
\end{equation}
The run-time is approximately doubled. If $\hat w_N$ and $\hat v_N$ are the MLE's then the estimates are $\hat D_{y}(q_1)=D_y(q_1;\hat w_N)$ and $\hat D_{x_1,y}(q_2)=D_{x_1,y}(q_2;\hat v_N)$. 

Finally, we plot the estimated distortion surface
\begin{equation}\label{eq:distortionsurface}
\hat d_{y_{obs}}(q_1,q_2)=\hat d_{G^{-1}_{y_{obs}}(q_1),y_{obs}}(q_2) \hat d_{y_{obs}}(q_1)
\end{equation}
as a diagnostic plot. Both components are simply Beta-densities and straightforward to evaluate.

\section{FURTHER RELATED WORKS}

\cite{prangle2014diagnostic} show that $\tilde\pi(x|y_{obs}) = \pi(x|y_{obs})$ for all $x$ iff $G_{y_{obs}}(X)\sim U(0,1)$ for $X \sim \pi(\cdot|y_{obs})$.
The authors give a diagnostic tool based on this idea for an ABC posterior using the simulated $Q$'s as test statistics. They sample $\{x_i, y_i\} $ from the truncated generative distribution $\pi(x)p(y|x)\mathbb{1}(y \in \Delta)$, where $\Delta \subset \mathcal{Y}$ is a subset containing $y_{obs}$, and compute $q_i = G_{y_i}(x_i)$ for $i = 1,...,N$. Then they check that the simulated $\{{q}_i\}_{i=1}^N$ are uniformly distributed over $[0,1]$. This corresponds to studying the distribution of the marginalized random variable $Q = E_{Y\in \Delta}(G_Y(X)|Y)$ rather than the conditional random variable $[Q|(Y=y_{obs})] = G_{y_{obs}}(X)$ which we study. The diagnostic histogram plotted by \cite{prangle2014diagnostic} estimates the marginal density $d_\Delta(\cdot)$ of $Q$, 
\begin{equation}\label{eq:marginalq}
    Q \sim d_\Delta(\cdot), \quad d_\Delta(Q) \propto \int_{y \in \Delta} d_y(Q) p(y) dy.
\end{equation}
Since $\Delta$ is typically rather large, $d_y(\cdot)$ may vary over $y\in \Delta$. In this case the marginal distribution of $Q$ may be flat when the conditional distribution of $Q|(Y=y_{obs})$ is far from flat (or the converse). We give an example in which this is the case. 
Similar ideas are explored in \cite{talts2018validating} and \cite{yao2018yes}. 
Notice that when we window our data $\{q,y\}, y\in \Delta$ for neural net regression estimation of $\hat w_N$ there is no integration over data $y$. We regress the distribution of  $Q|(Y=y)$ at each $y\in \Delta$ (i.e. close $y_{obs}$), so we explicitly model variation in $d_y(\cdot)$ with $y$ within $\Delta$. 

\cite{rodrigues2018recalibration} give a post-processing recalibration scheme for the ABC posterior developing \cite{prangle2014diagnostic}. The setup is a multivariate version of Equation~\ref{eq:dmap}. Ignoring the intrinsic ABC approximation, the ``approximation'' they correct is due to the fact that they have posterior samples at one $y$-value and they want to transport or recalibrate them so that they are samples from the posterior at a different $y$-value.  
The main difference is that these authors are approximating the true posterior, whilst we are trying to avoid doing that.

\cite{2019arXiv190507488G} propose Automatic Posterior Transformation (APT) to construct an approximate posterior. Our Algorithm~1 can be seen as the first loop of their Algorithm~4. In their notation, let $q_{F(y,w)}(x)$ be an approximation to $\pi(x|y)$ where $w$ are parameters of the fitted approximation. Let $p_r(x)$ be a proposal distribution for $x$. Define 
\begin{equation}
    \tilde q_{F(y,w)}(x) = q_{F(y,w)}(x) \frac{p_r(x)}{\pi(x)Z(y,w)},
\end{equation}
with $Z$ a normalisation over $x$, and 
\begin{equation}
    \tilde{\mathcal{L}}(w) = \sum_{i=1}^N \log \tilde q_{F(y_i,w)}(x_i),
\end{equation}
where $\{x_i,y_i\} \sim p_r(x)p(y|x)$ iid for $i=1,...,N$. Appealing to \citet{papamakarios2016fast}, the authors show that the $w$-values maximising the scaled limit of $\mathcal{L}(w)$, $w^*$ say, satisfy $q_{F(y,w^*)}(x) =  \pi(x|y)$ (if the representation is sufficiently expressive) and this leads to a novel algorithm for approximating the posterior. Our approach is a special case obtained by taking $x \in \mathcal{R}$, $p_r(x) = \pi(x)$ (so $Z=1$) and the special parameterization
\begin{equation}
    q_{F(y,w)}(x) = d_y({G}(x);w)\tilde\pi(x|y).
\end{equation}
In further contrast, we are concerned with diagnosing an approximation, not targeting a posterior. 

Some previous work on diagnostics has also avoided forming a good approximation to $F_Y$ by focusing on estimating the error for expectations of special functions only. Work on calibration of credible sets by \citet{pmlr-v97-xing19a} and \citet{lee2018calibration} falls in this category. Instead of estimating distortion over the whole CDF $G_Y$, these authors estimate the distortion in the value of one quantile. This lacks diagnostic detail compared to our distortion map. They consider a level $q$ approximate credible set $\tilde C_y(q)\subseteq \mathcal{X}$ computed from the approximate posterior. 
\citet{pmlr-v97-xing19a} estimate how well this approximate credible set covers the true posterior, that is they estimate 
\begin{equation}\label{eq:coverage}
    c_{y_{obs}}(q)= E_{X|Y=y_{obs}}(\mathbb{1}(X\in \tilde C_{y_{obs}}(q))),
\end{equation}
using regression and methods related to importance sampling. In contrast to \cite{pmlr-v97-xing19a}, we estimate $D_{y_{obs}}(q)$ as a function of $q$, so we estimate the distortion in the CDF, not just the distortion in the mass it puts on one set. 

\section{TOY EXAMPLE}\label{sec:logistic}

\begin{figure}[t!]%
    \centering
    \subfloat{{\includegraphics[width=0.22\textwidth]{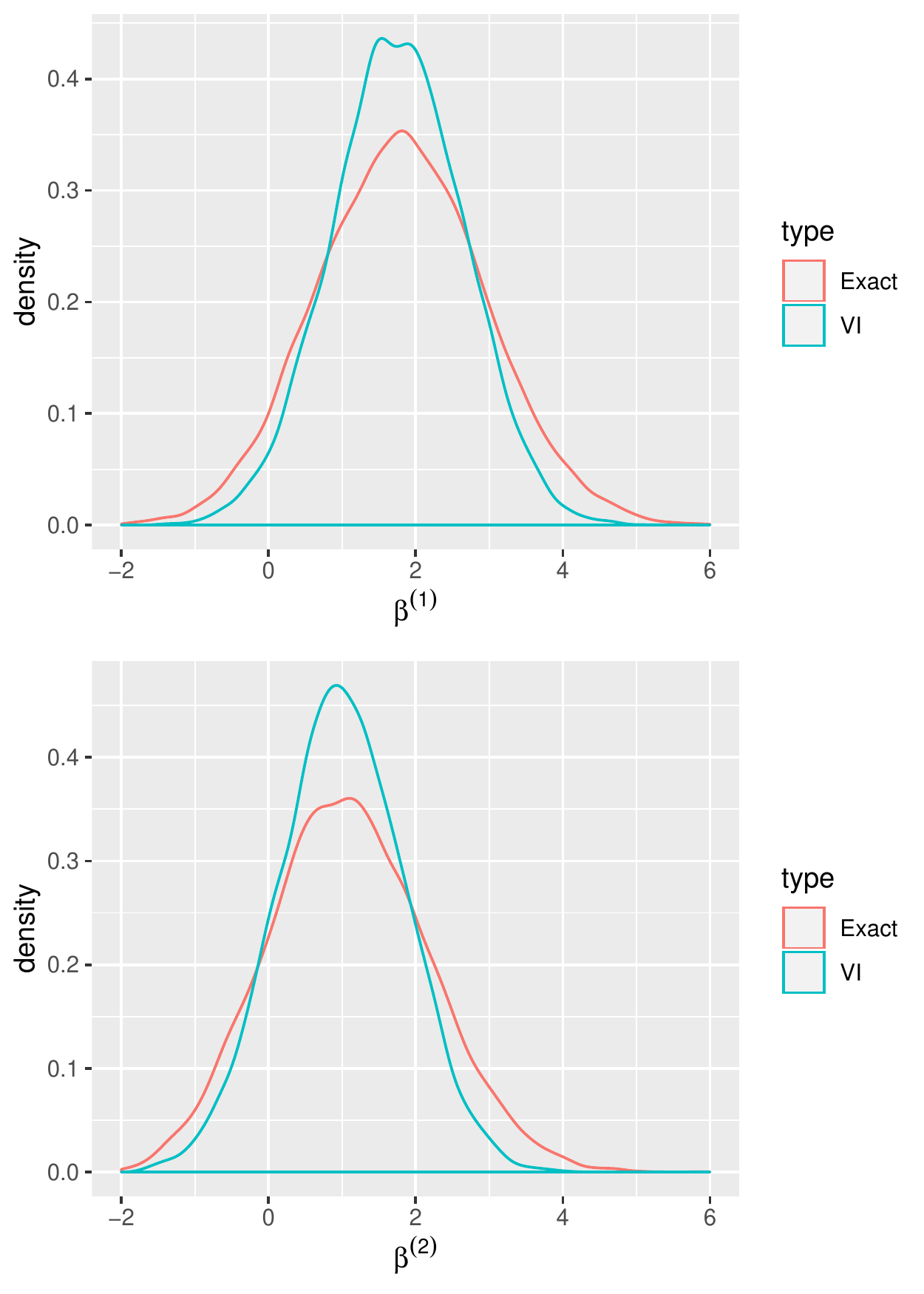} }}
    \subfloat{{\includegraphics[width=0.22\textwidth]{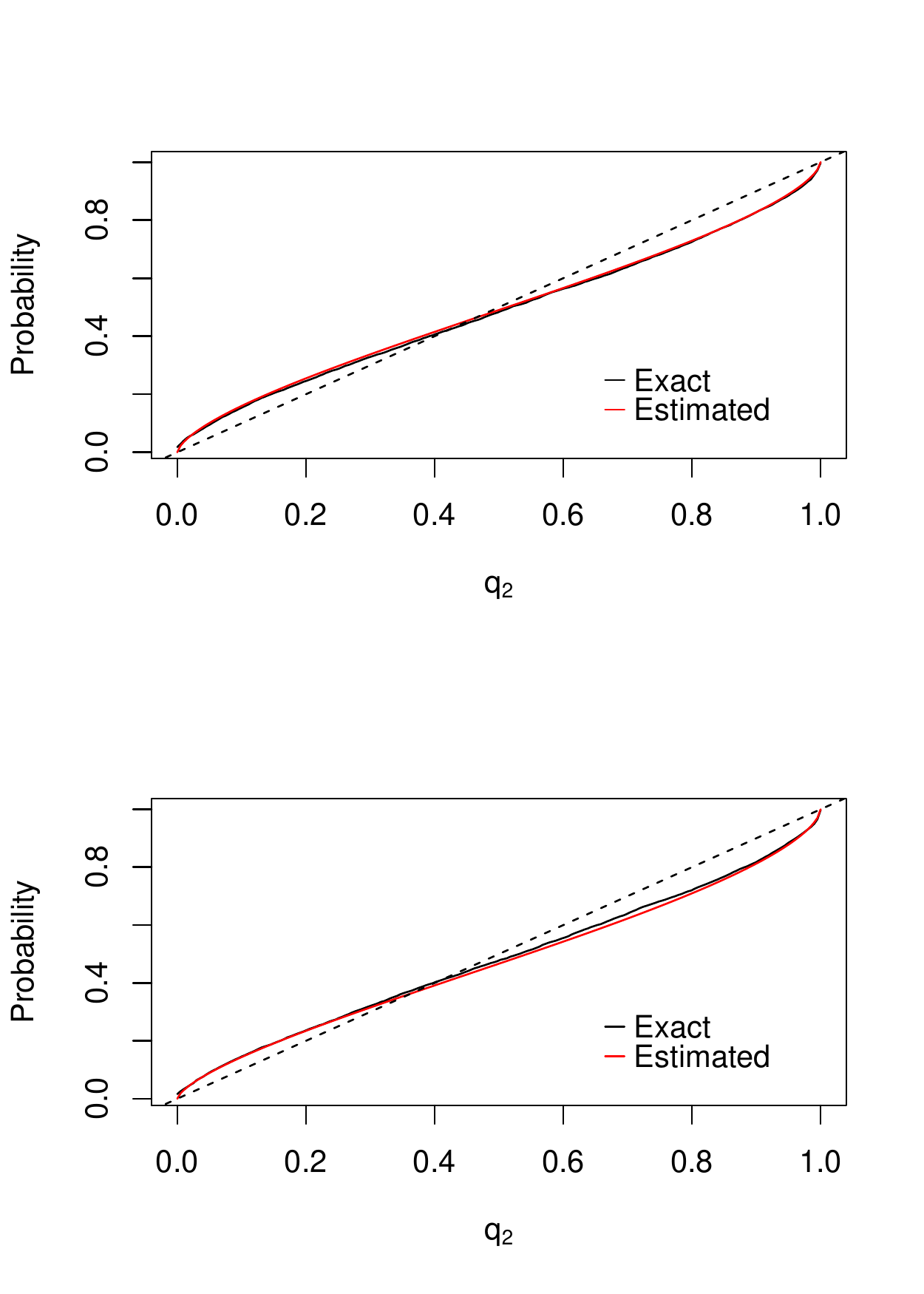} }}%
    \caption{Left: Exact and approximate posterior for rows $\beta^{(i)}$, $i = 1,2$. Right: Exact $D^{(p)}_{y_{obs}}(\cdot)$ and fitted $\hat{D}_{y_{obs}}^{(i)}(\cdot)$ for $\beta^{(i)}$, $i = 1,2$.  Dashed line is the identity map.}%
    \label{fig:logitdistortion1}%
\end{figure}

We apply Algorithm~1 to Bayesian logistic regression. Let $X$ be a $n\times p$ design matrix, let $\beta \in \mathcal{R}^p$ be regression coefficients and $y = (y_{(1)},...,y_{(n)}) \in \{0,1\}^n$ be binary response data. For each $j=1,...,n$, $y_{(j)} \sim \mbox{Bernoulli}(p_j)$ where $\mbox{logit}(p_j) = x_{(j)}^T\beta$ and $x_{(j)}$ is the $j$th row of $X$. The likelihood is \[p(y|\beta) = \prod_{j=1}^n {p_j}^{y_{(j)}}{(1-p_j)}^{1-y_{(j)}}, \ \ p_j = \frac{\exp (x_{(j)}^T \beta)}{\exp (x_{(j)}^T \beta) + 1}\]
We take a prior distribution $\pi(\beta) = \mbox{Normal}(0,2I_p)$ with $I_p$ the $p \times p$ identity matrix. We are interested in the posterior distribution $\pi(\beta|y_{obs}) \propto \pi(\beta)p(y_{obs}|\beta)$. 

The exact posterior can be sampled via standard MCMC.
It is also possible to approximate the exact $\pi(\beta|y_{obs})$ using computationally cheaper Variational Inference (VI) with posterior $\tilde \pi(\beta|y_{obs})$ \citep{jaakkola1997variational}. In this example, we set $p=8,n=50$, and we would like to diagnose the performance of the variational posterior $\tilde \pi(\beta|y_{obs})$ using Algorithm~1. In our example, each entry in the design matrix $X$ is sampled independently from $U(0,1)$. We simulate $10^6$ synthetic $\{\beta,y\}$-pairs from the generative model $\pi(\beta)p(y|\beta)$, randomly pick one synthetic data point as our observed $y_{obs}$, and keep the $1\%$ of pairs $\{\beta_i,y_i\}_{i=1}^N$ closest in Euclidean distance to $y_{obs}$ as our training data (this corresponds to a particular choice of $\Delta$ in Algorithm~1). Since there is no low dimensional sufficient statistic for this model, we simply use $s(y)=y$, the $n=50$ dimensional binary response vector, as the summary statistic. We then apply Algorithm~1 using a feed froward neural net with two hidden layers of 80 nodes to estimate the distortion map $\hat D^{(j)}_{y_{obs}}(\cdot)$ for each dimension $j = 1,...,p$ of $\beta$ (recall $p=8$), and compare the estimated map $\hat D^{(j)}_{y_{obs}}(q)$ to the exact $D^{(j)}_{y_{obs}}(q)$ as a function of $q\in [0,1]$ (the exact map is available for this problem using standard methods). 

We plot the marginal posteriors and the corresponding exact and fitted distortion maps for the first two dimensions $\beta^{(1)}$, $\beta^{(2)}$ of the regression parameter $\beta$ in Fig.~\ref{fig:logitdistortion1}. The fitted distortion map $\hat D^{(1)}_{y_{obs}}(\cdot)$ and $\hat D^{(2)}_{y_{obs}}(\cdot)$ in the right column accurately recover the exact map. Both $\hat D^{(1)}_{y_{obs}}(\cdot)$ and $\hat D^{(2)}_{y_{obs}}(\cdot)$ slightly deviate from the identity map, correctly showing that the marginal VI posteriors for $\beta^{(1)}$ and $\beta^{(2)}$ are slightly under-dispersed compared to the exact posterior. This simple example shows our method is able to handle moderately high dimensional ($n=50$) summary statistics $s(y)$.

\section{KARATE CLUB NETWORK} \label{sec:karate}
In this section we estimate distortion maps measuring the quality of three distinct network model approximations. The data we choose are relatively simple, but happen to illustrate several points neatly. We repeat the analysis on a larger data set in Appendix. The small size of the network data in this example is not an essential point. Conclusions from the larger data set are similar though in some respects less interesting. 

The Zachary's Karate Club network \citep{zachary1977information} is a social network with 34 vertices (representing club members) and 78 undirected edges (representing friendship). The data is available at \href{http://vlado.fmf.uni-lj.si/pub/networks/data/ucinet/ucidata.htm#zachary}{UCINET IV Datasets}. See Fig.~\ref{fig:karate}.

\begin{figure}[t!]
    \centering
    \includegraphics[width=0.33\textwidth]{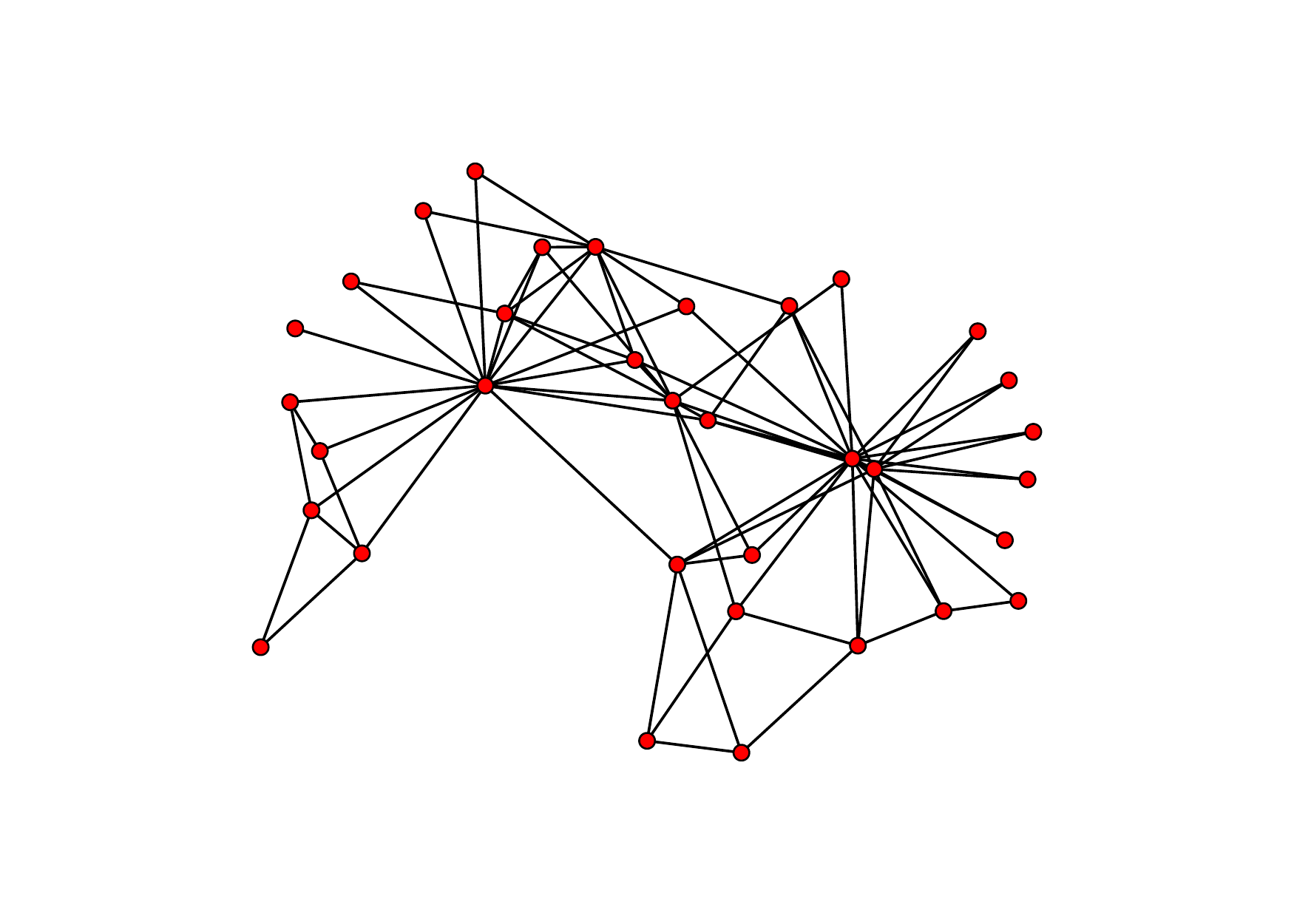}
    \caption{Zachary's Karate Club network \citep{zachary1977information}, consists of 34 vertices and 78 undirected edges.}
    \label{fig:karate}
\end{figure}


We fit an Exponential Random Graph Model (ERGM) \citep{robins2007introduction} to these data. Let $\mathcal{Y}$ be the set of all graphs with $n$ nodes. Given $y \in \mathcal{Y}$, let $s(y) \in \mathcal{R}^p$ be a $p$-dimensional graphical summary statistic computed on $y$ and let $x \in \mathcal{R}^p$ be the corresponding ERGM parameter. In our example $p = 3$. In an ERGM, the likelihood of the graph $y$ is 
\begin{equation}\label{eq:ergm}
p(y|x) = \exp{ \{x^Ts(y)\}}/z(x)
\end{equation}
where $z(x) = \sum_{y \in \mathcal{Y}} \exp{ \{x^Ts(y)\}}$ is intractable even for relatively small networks.

Our example approximations come from \cite{caimo2012bergm} and \cite{bouranis2018bayesian}. Let $s_1(y)$ be the number of edges in $y$. Following \cite{hunter2006inference}, let $\mbox{EP}_l(y)$ be the number of connected dyads in $y$ that have $l$ common neighbors, and let $D_l(y)$ equal the number of nodes in $y$ that have $l$ neighbors. Let
\[
v(y,\phi_v) = e^{\phi_v} \sum_{l=1}^{n-2}\{1-(1-e^{-\phi_v})^l\}\mbox{EP}_l(y)
\]
be the geometrically weighted edgewise shared partners (gwesp) statistic and
\[
u(y,\phi_u) =  e^{\phi_u} \sum_{l=1}^{n-1}\{1-(1-e^{-\phi_u})^l\}\mbox{D}_l(y)
\]
be the geometrically weighted degree (gwd) statistic. Following \cite{caimo2012bergm} let $\phi_v = 0.2$ and $\phi_u = 0.8$, $s(y) = (s_1(y),v(y,\phi_v),u(y,\phi_u))$ and $x = \{x^{(1)},x^{(2)},x^{(3)}\} \in \mathcal{R}^3$.
Our observation model is given by Eqn~\ref{eq:ergm}.
The prior distribution $\pi(\cdot)$ for $x$ is multivariate normal with $\mu = (-2,0,0)$ and $\Sigma = 5I_3$. 
 
 The exact $\pi(x|y)= \pi(x) p(y|x)/p(y)$ is doubly intractable. We consider three approximation schemes yielding different approximations $\tilde{\pi}(x|y)$: 
 \vspace{-3pt}
\begin{itemize}
    \item Approximate Bayesian Computation with ABC acceptance fraction $\rho = 0.5\%$ and local linear regression adjustment (``ABC-reg'', \cite{pritchard1999population,beaumont2010approximate}). 
    \item Fully adjusted pseudolikelihood (``adj-lkd'') \citep{bouranis2017efficient, bouranis2018bayesian};
    \item Variational inference (VI) \citep{tan2018bayesian}
\end{itemize} 
\begin{figure}
    \centering
    \includegraphics[width = 0.44\textwidth]{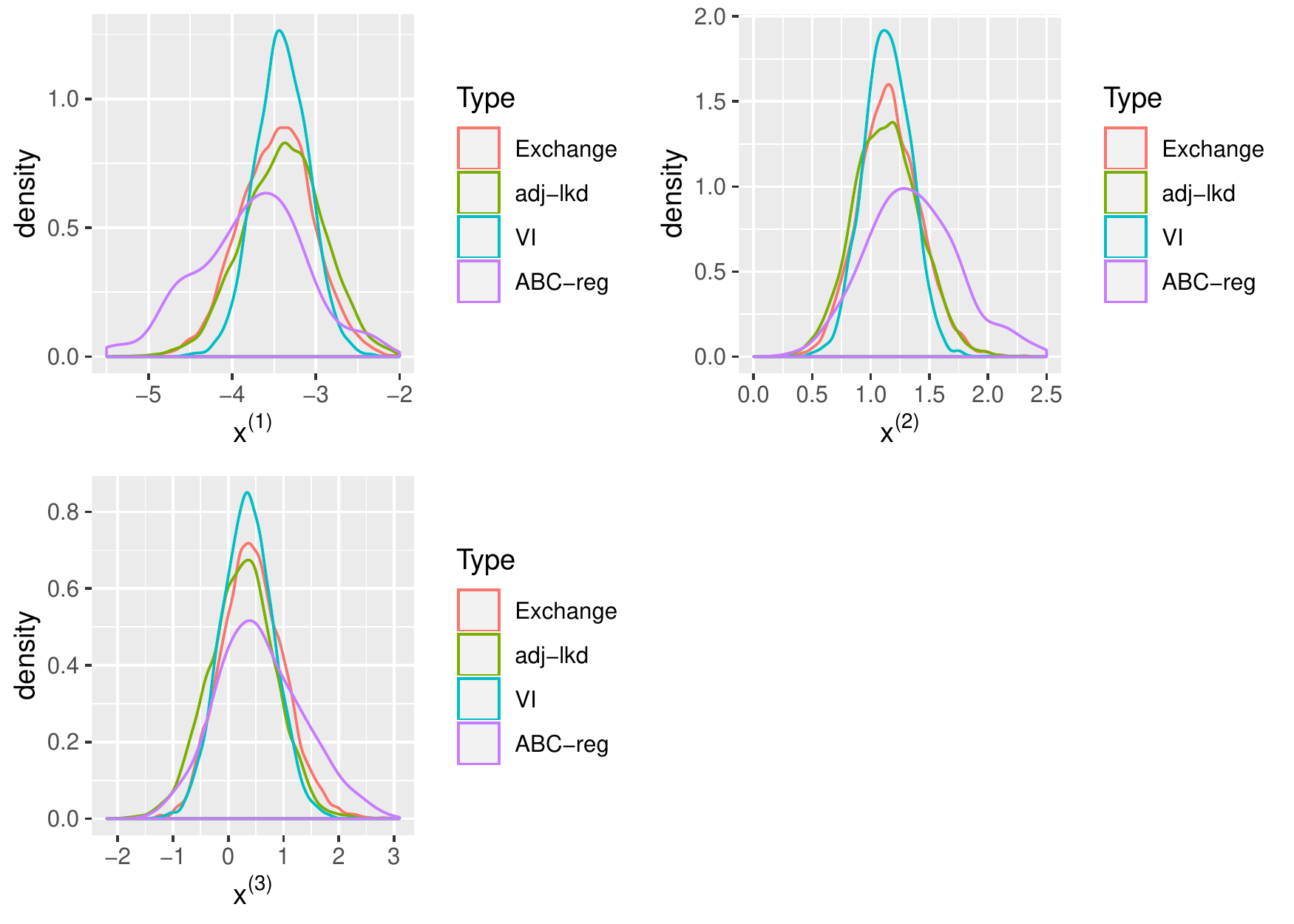}
    \caption{Approximate and exact posteriors}
    \label{fig:karateapprox}
\end{figure}

\vspace{-4pt}
We have ground truth in this example, sampling $\pi(x|y)$ using an approximate exchange algorithm \citep{murray2012mcmc}. This is still approximate but very accurate. 
For each approximation scheme and dimension $x^{(p)}$, $p=1,...3$, we fit the distortion map $\hat{D}^{(p)}_{y_{obs}}$ using Algorithm~1 and compare our $\hat{d}^{(p)}_{y_{obs}}$-diagnostic plot with diagnostic plots obtained using the methods of \cite{prangle2014diagnostic} and \cite{talts2018validating}.


We simulated $N=3 \times 10^5$ pairs $\{x_i,y_i\}_{i=1}^N$ from the generative model $\pi(x)p(y|x)$, taking pairs $\{x_i,y_i\}$ pairs in the top $15\%$ by least Euclidean distance to $s(y_{obs})$ as our training data. We first report the approximate posteriors themselves. In Fig.~\ref{fig:karateapprox} (left column) we see that the adj-lkd approach (top row) gives the best approximate posterior for all dimensions. In comparison, the VI approach (bottom row) gives an under-dispersed approximation while the ABC-reg posterior (middle row) is over-dispersed and slightly biased. In a real application we would not have this ground truth.
\begin{figure}%
    \centering
    \subfloat{{\includegraphics[width=0.23\textwidth]{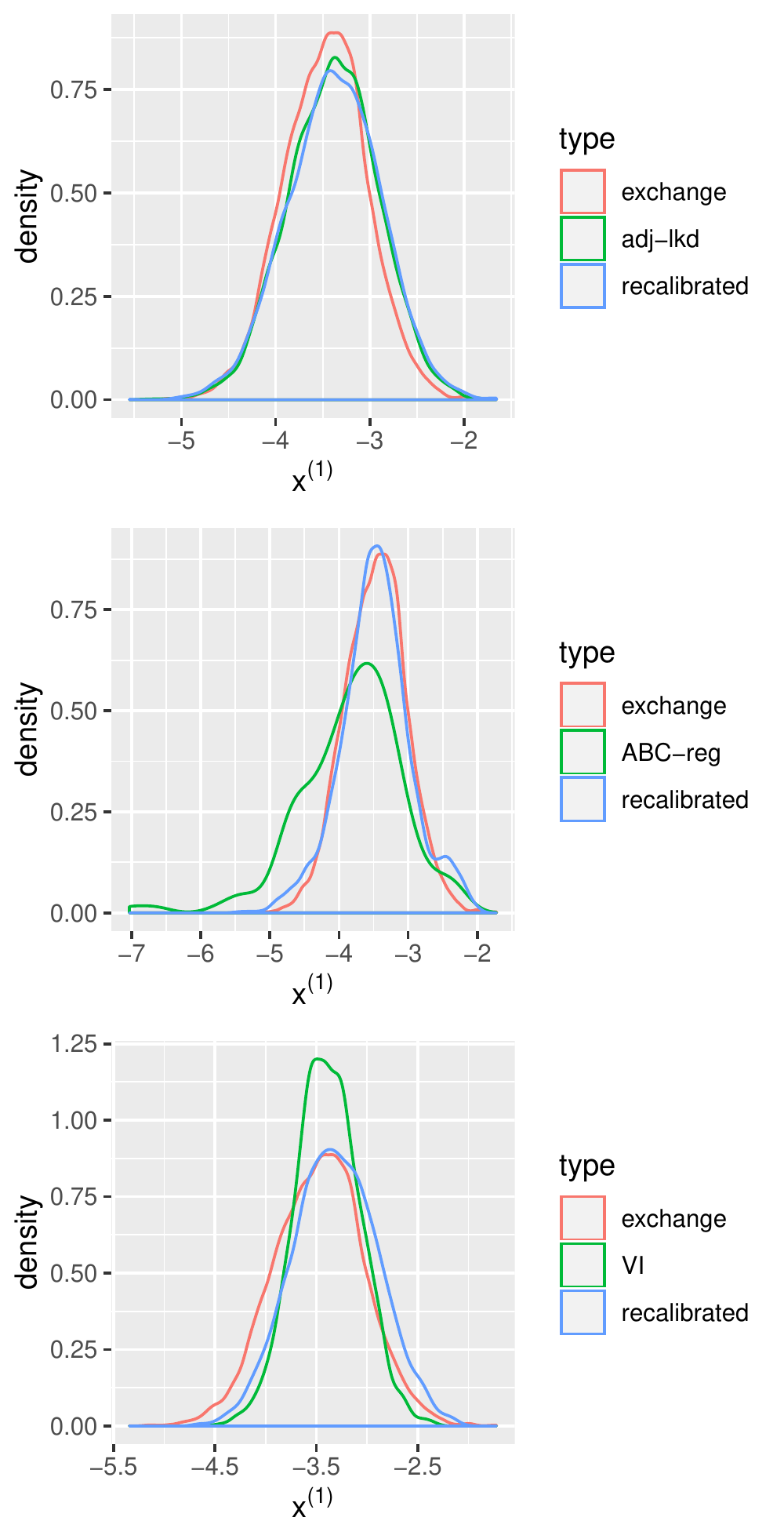} }}
    \subfloat{{\includegraphics[width=0.23\textwidth]{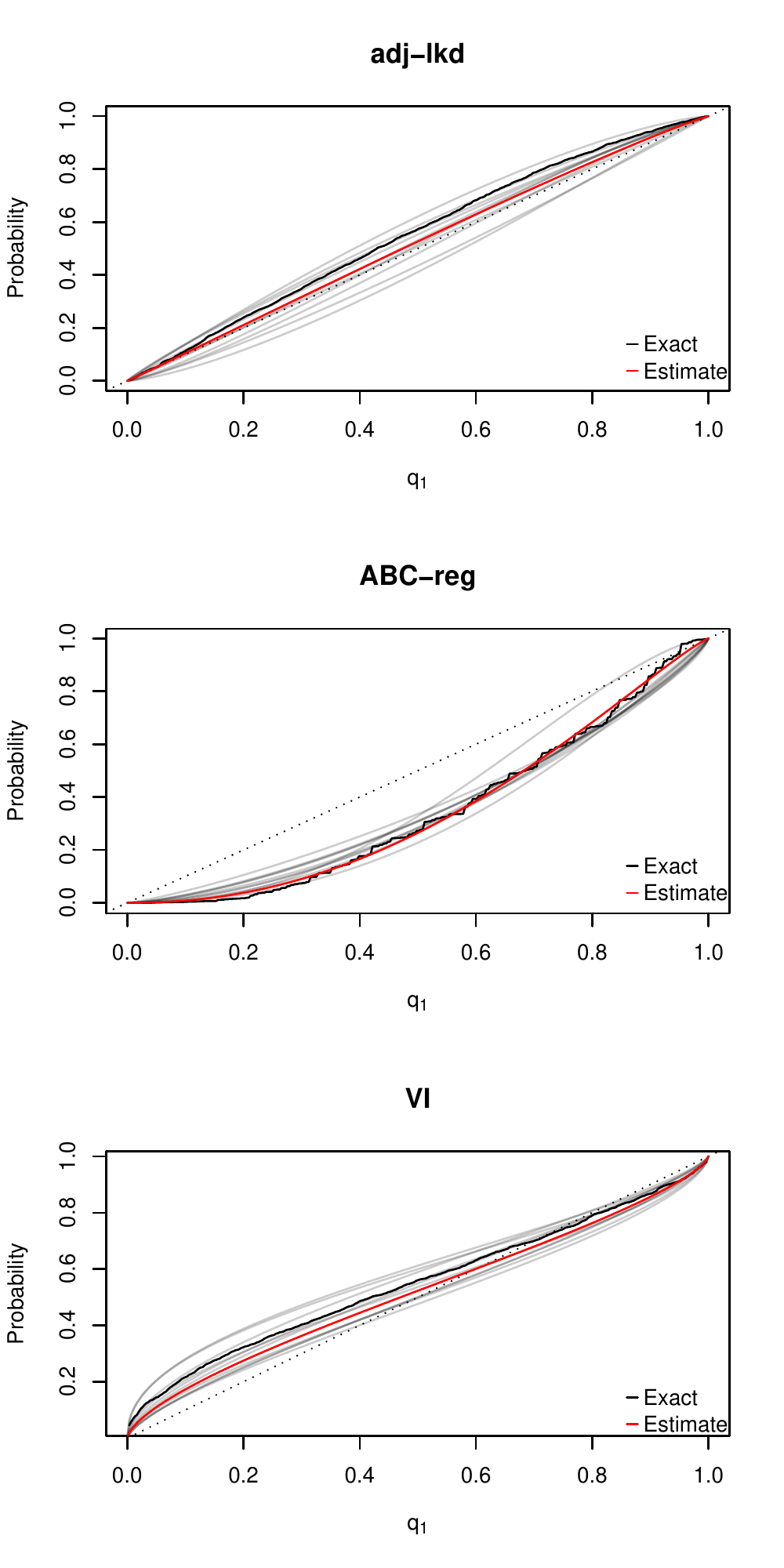} }}%
    \caption{Left: Recalibrated posterior $\hat F_{y_{obs}}$ for $x^{(1)}$ for each approximation scheme Right: Exact $D^{(1)}_{y_{obs}}(\cdot)$ and fitted $\hat{D}_{y_{obs}}^{(1)}(\cdot)$ for $x^{(1)}$,  Dashed line represents the identity map. Grey lines are $\hat D^{(1)}_{y_{obs}}(\cdot)$ fitted repeatedly using $70\%$ random subset of the training data.}%
    \label{fig:karatedistortion1}%
\end{figure}

We now run Algorithm~1 using a feed froward neural net with two hidden layers of 80 nodes and estimate $\hat{D}^{(p)}_{y_{obs}}$ for all approximation schemes and dimensions $x^{(p)}$. For brevity we now focus on the distribution of $x^{(1)}$. In Fig.~\ref{fig:karatedistortion1} (right column) we show the exact $D_{y_{obs}}^{(1)}(\cdot)$ (not available in real applications, but useful to show the method is working) and the fitted $\hat D_{y_{obs}}^{(1)}(\cdot)$ for $x^{(1)}$ for all three approximation schemes with the corresponding recalibrated posteriors $\hat \pi(x^{(1)}|y)$ (left column). 
For all approximation schemes the estimated $\hat D_{y_{obs}}^{(1)}(\cdot)$ is close to the exact $D_{y_{obs}}^{(1)}(\cdot)$, and are stable under repeated runs (which were fitted using $70\%$ of the training data). The approximate posterior (with CDF $G_y$) matches the true posterior (in the graphs at left in Fig.~\ref{fig:karatedistortion1}) when $\hat D_{y_{obs}}^{(1)}(\cdot)$ is close to an identity map (in the graphs at right in the same figure). The recalibrated posteriors (with CDF $\hat F_y$) are closer to the exact, again indicating that our fitted $\hat D_{y_{obs}}^{(1)}(\cdot)$ is correct. 

\begin{figure}[t]
    \centering
    \includegraphics[width = 0.45\textwidth]{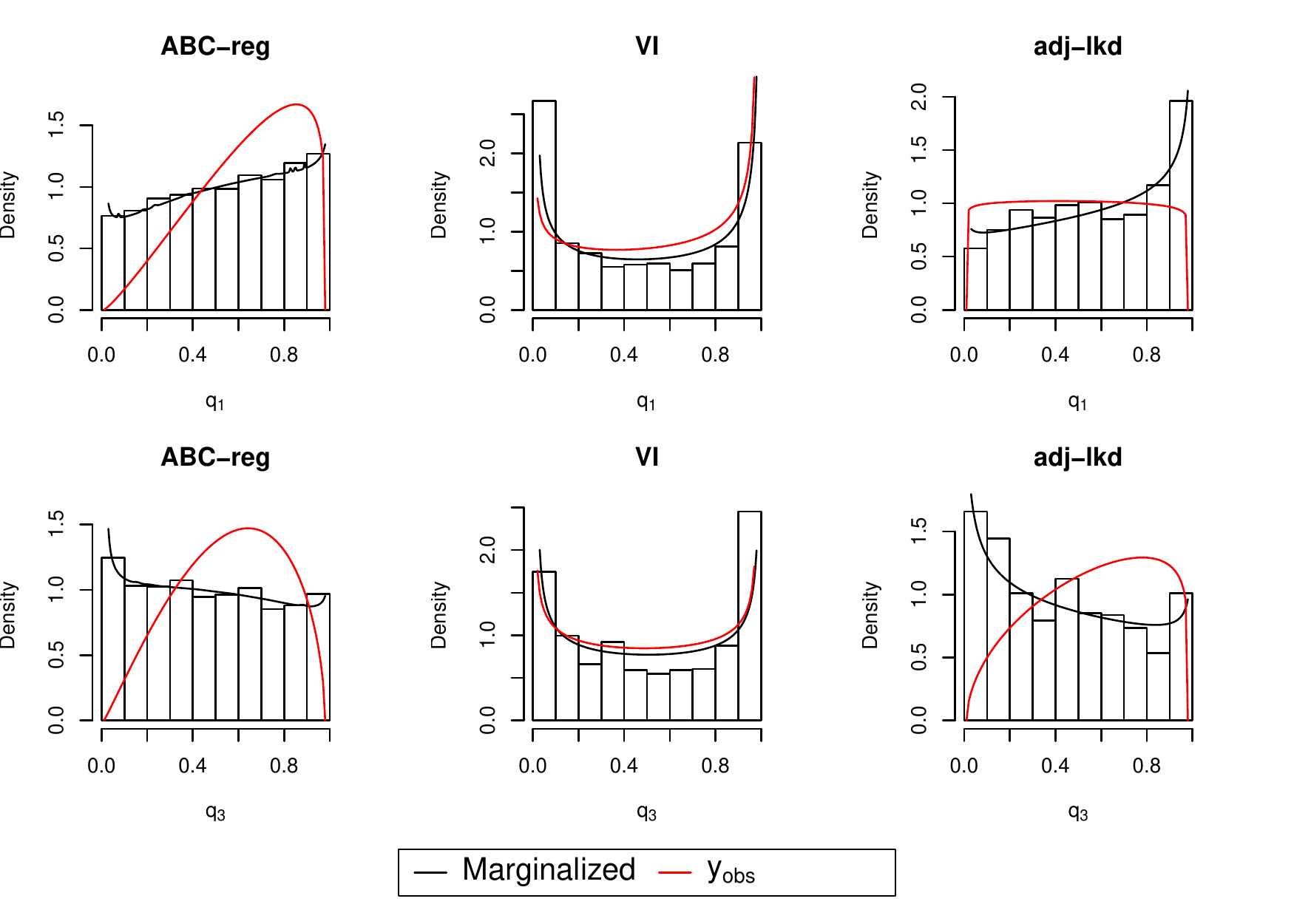}
    \caption{Diagnostic plot \citep{prangle2014diagnostic} for each approximation scheme for $x^{(1)}$ (upper) and $x^{(3)}$ (lower). Black curve: marginalized (averaged) $\hat{d}_\Delta(\cdot)$ over $y$ s.t. $s(y) \in \Delta_{s(y_{obs})}$. Red curve: fitted $\hat d_{y_{obs}}(\cdot)$ at $y_{obs}$. Recall that $\hat{d}(\cdot)$ represents the corresponding PDF of $\hat{D}(\cdot)$} 
    \label{fig:karatediaghist}
\end{figure}
\vspace{-3pt}
Plots of the distortion density $d_y$ allow direct comparison with the diagnostic histograms of \cite{prangle2014diagnostic} and \cite{talts2018validating}. Adopting those methods in our setting, we average $d_y$ over an open ball centered at $s(y_{obs})$ containing the top $2.5\%$ of $s(y_i)$'s closest to $s(y_{obs})$. We see in Fig.~\ref{fig:karatediaghist}, where we plot diagnostics for $x^{(1)}$ (top row) and $x^{(3)}$ (bottom row), that these diagnostic histograms successfully identify the under-dispersion of VI posteriors (a U-shape in the corresponding histograms \citep{talts2018validating} in the middle column). However, the histogram of ABC-reg is reasonably flat for $x^{(3)}$, which seems healthy. This is misleading as the ABC-reg posterior for $x^{(3)}$ is in fact over-dispersed at $y_{obs}$ as the $d_y$-graph in red shows. In contrast, the non-uniformity in the histogram of adj-lkd posterior of $x^{(1)}$ (top right) suggests that that approximation is poor, when we see from $d_y$-graph in red that the approximation is excellent (with ground truth in the top row of Fig.~\ref{fig:karateapprox} agreeing). The diagnostic histograms of  \cite{prangle2014diagnostic} and \cite{talts2018validating} 
give both false-positive and false-negative alerts in this example.

To further illustrate this behavior on the adj-lkd example for $x^{(1)}$, we sampled $K = 200$ pairs $\{x_k,y_k\}_{k=1}^K \sim \pi(x)p(y|x)\mathbb{1}(s(y) \in \Delta_{s(y_{obs})})$, so that the $s(y_k)$'s are all close to $s(y_{obs})$. For each data set $y_k$, we compute an equal-tail approximate credible set with level $\alpha=0.8$ for $x^{(1)}$ using the adj-lkd posterior. Following \cite{pmlr-v97-xing19a} we can ask, what is the true (i.e. ``operational'') coverage $\tilde c^{(1)}_{y_k}(\alpha)$ achieved by this approximate set in the exact posterior? Does the approximate credible set at the data have the stated coverage in the true posterior? The exchange algorithm gives (fairly accurate) samples from the true posterior so the expectation in Eqn.~\ref{eq:coverage} is easily estimated. 

In Fig.~\ref{fig:coverage_visual} we plot the points $s(y_k)\in \mathcal{R}^3$ colored by their coverage. Red points points correspond to data where we are getting the right coverage.
However there is an orange-colored plane region in the top right part of the plot where $\tilde c^{(1)}_{y_k}(\alpha)$ is much lower than the nominal level of $80\%$. The data $y_{obs}$ is located at a red point so the coverage from the adj-lkd approximation is fine (as we would expect from the healthy diagnostics in Fig~\ref{fig:karatediaghist}). However when we average we include data where the approximation is poor and reach the wrong conclusion. This illustrates how the quality of approximation can vary over a subset of data space $\mathcal{Y}$. 

\begin{figure}[t!]
    \centering
    \includegraphics[width = 0.44\textwidth]{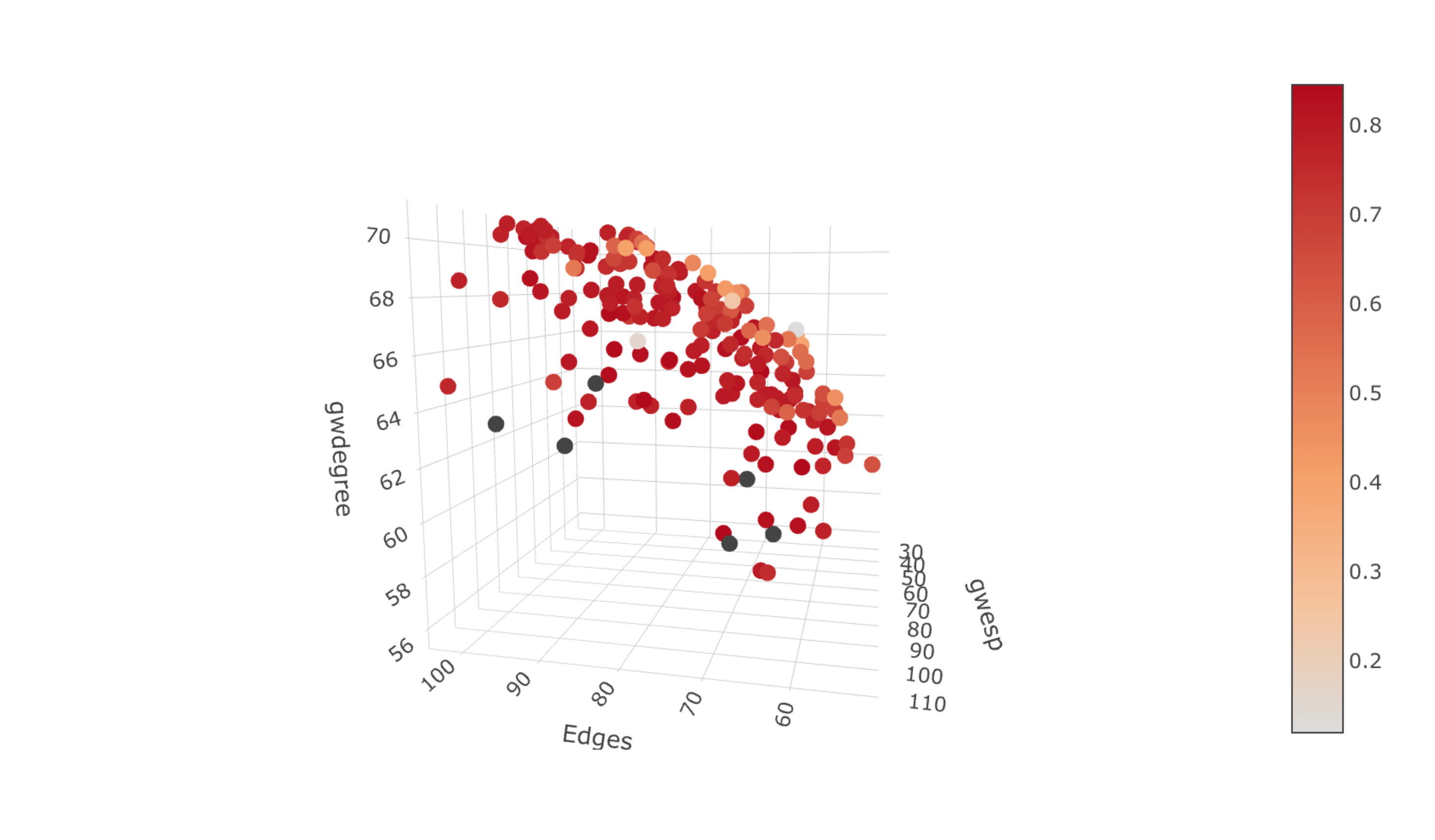}
    \caption{The estimated operational coverage of adj-lkd posterior of $x^{(1)}$ at each $s(y)$, magnitude of operational coverage is represented by colour, nominal level $\alpha = 0.8$}
    \label{fig:coverage_visual}
\end{figure}

Finally, we estimate and report the bivariate distortion surface $d_{y_{obs}}$ for
VI and adj-lkd approximations $\tilde{\pi}(x^{(1)},x^{(2)}|y_{obs})$ 
to the posterior for the first two parameters $x^{(1)}$ and $x^{(2)}$. From Sec.~\ref{sec:mv}, taking $q_1 = G_{y_{obs}}(x^{(1)})$ and $q_2 = {G}_{x^{(1)},y_{obs}}(x^{(2)})$, the distortion surface $d_{y_{obs}}(q_1,q_2)$ is
\[
d_{y_{obs}}(q_1,q_2)\equiv d_{G^{-1}_{y_{obs}}(q_1),y_{obs}}(q_2) d_{y_{obs}}(q_1).
\]
Fig.~\ref{fig:karatesurface} shows that for the VI posterior, the distortion surface peaks on the boundary and corners of the $[0,1]^2$ square, and is below 1 at the center (recall that it is a normalised bivariate probability density). This is the 2-D equivalent of the U shaped diagnostic plots for scalars described in \cite{prangle2014diagnostic} and \cite{talts2018validating}, reflecting the under-dispersed VI posterior approximation. In contrast, the distortion surface of adj-lkd posterior is between $0.9 \sim 1.2$ and relatively flat over much of the $[0,1]^2$ square: there is no evidence here for a problem with the adj-lkd approximation. 

\section{CONCLUSION}
In this paper we give new diagnostic tools for approximate Bayesian inference. The distortion map $D_{y_{obs}}$ is a visual diagnostic tool for approximate marginal posteriors, which gives us diagnostic details about the approximation error. It is computationally demanding to estimate. Estimating the distortion map $D_{y_{obs}}$ requires sampling synthetic data from the generative model and calling the approximation scheme at each synthetic data point. In contrast to existing methods it checks the quality of approximation \emph{at the observed data} $y_{obs}$, instead of estimating ``averaged performance" over data space.
Much of the code-base (simulation outline, fitting the Beta-density
conditioned on $y$-values in a neighborhood of $y_{obs}$) carries over from one problem to another, so the user provides simulators for the generative model and the approximate posterior. The approach can be extended from diagnosing univariate marginals to higher dimensions. One interesting direction for future work is to find a way to simulate synthetic data close to $y_{obs}$ while reweighting in a way that yields an unbiased distortion map.

\begin{figure}[t!]%
    \centering
    \subfloat{{\includegraphics[width=0.53\linewidth]{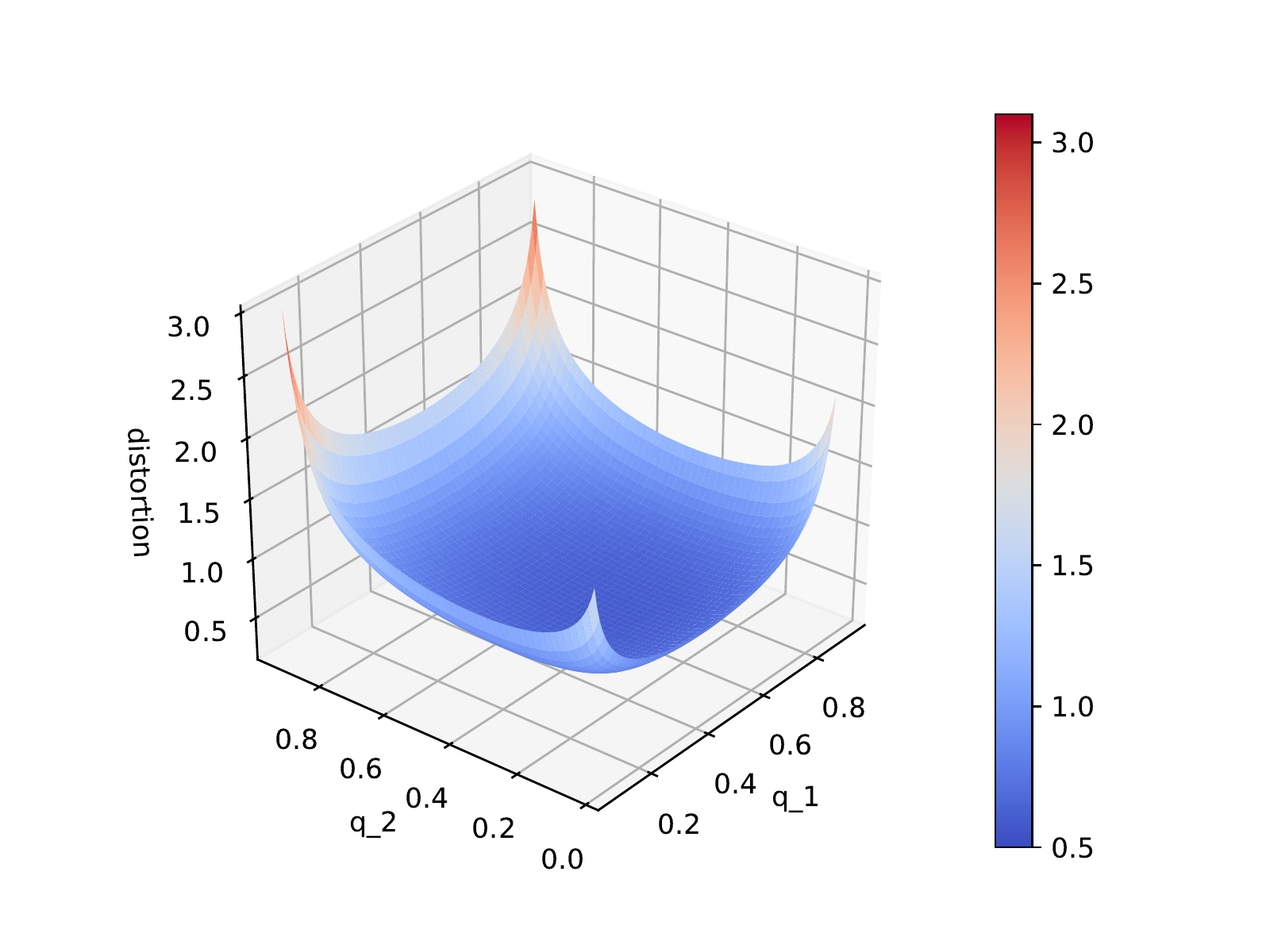} }}
    \subfloat{{\includegraphics[width=0.53\linewidth]{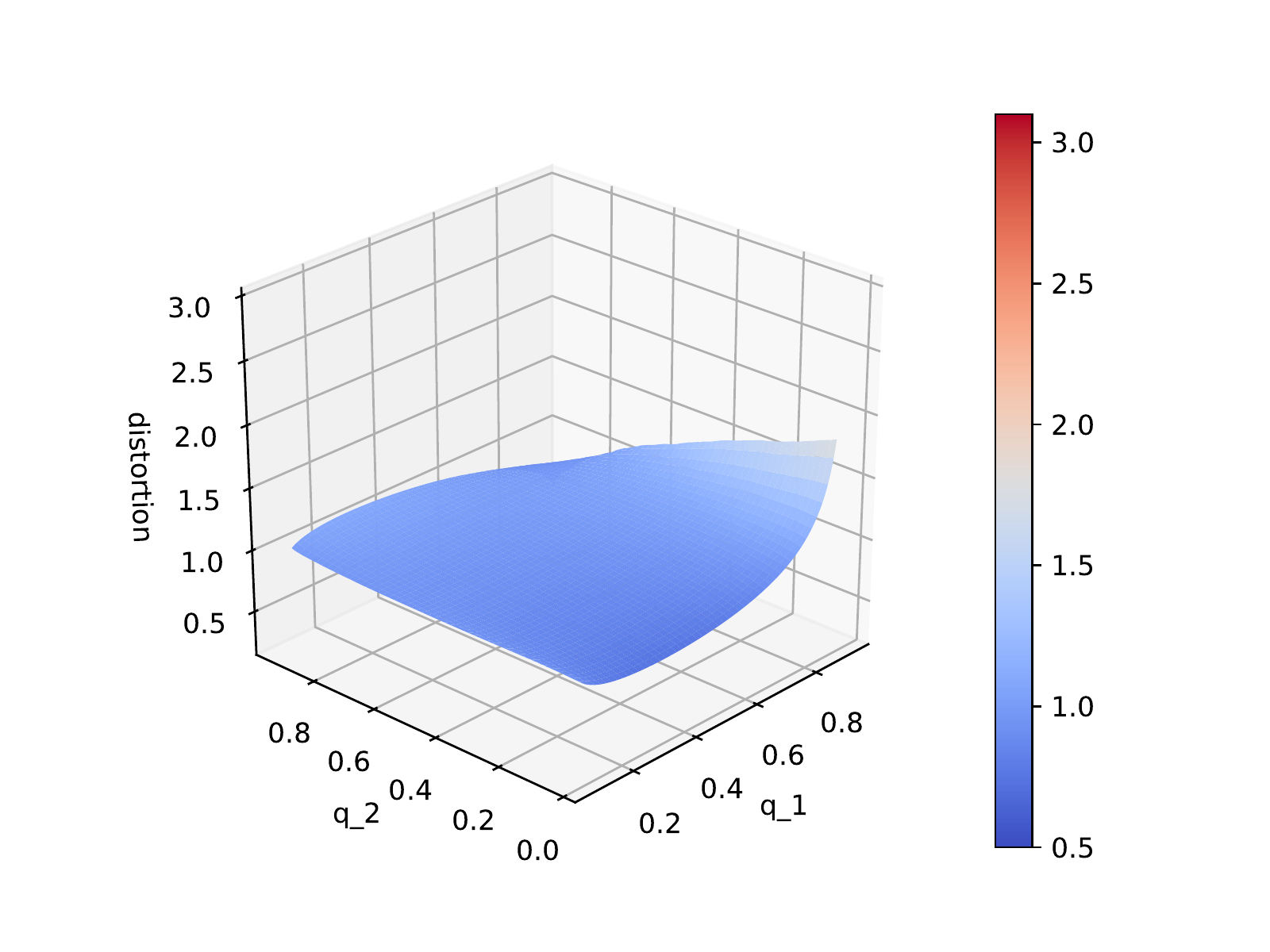} }}%
    \caption{Left: Distortion surface of VI posterior with respect to $q_1,q_2$. Right: Distortion surface of adj-lkd posterior with respect to $q_1,q_2$.}%
    \label{fig:karatesurface}%
\end{figure}

\bibliographystyle{apalike}

\newpage
\appendix
\section{APPENDIX}

\subsection{Further discussion}

\subsection*{Comparison to ABC}
Distortion map estimation shares a number of features with ABC. These include simulation of the generative model and the presence of windowing on data $y\in\Delta$. The window plays different roles, as a marginalising window in ABC recalibration and a conditioning window in distortion map analysis but seems superficially similar. How do the methods compare?

Compared to standard ABC, estimation of a distortion map is fundamentally easier. This is illustrated in Section~\ref{sec:logistic}
where ABC is clearly overdispersed but the distortion map is accurately estimated.
ABC sets out to approximate the entire joint distribution of the
multivariate parameter. For diagnostic purposes the distortion map is targeting scalar or at most bivariate marginals only. The regularisation allowed by the restricted parameterisation of $D_y$ (see the next subsection) is helpful also.

Compared to ABC, estimation of the distortion map has the additional computational cost of a) repeatedly applying the approximation scheme on each synthetic data points (easy when $G_y$ is available in closed form, as is sometimes the case, as in mean-field VI). Existing methods \citep{talts2018validating,rodrigues2018recalibration} pay the same price. Another cost is b) fitting the network to the simulated data set (we do this just once). In our experience b) requires much less time than a), so the method presented in this paper works best when the approximation scheme is computationally cheap.

\subsubsection*{Parameterisation of $D_y$}
In this paper we parameterise $D_y$ as a Beta CDF. This may seem an arbitrary and restrictive choice. However we are partly benefiting from the normalising-flow parameterisation we have set up, as the distortion map is a CDF on $[0,1]$. More fundamentally we feel that a parametric restriction or ``regularisation'' of this sort is the price we pay for estimating a bias (i.e. the distortion of the approximate posterior from the true) without knowing the truth. We are using the Neural Net to regress on a space of (scalar) functions $D_y(x)$. By restricting this function space we regularise the fit in a helpful way. Other (possibly more flexible) parameterisations of $D_y$ are available. For example, we tried parameterising $D_y$ with a mixture of Beta CDFs (up to 4 components) but found no improvement, just longer run times.

One drawback of our setup is that the single component Beta can be fooled: for example, if the true distortion density $d_y$ was trimodal with peaks at 0.01, 0.5 and 0.99, then the estimated $\hat d_y$ would be close to uniform over [0,1] so our estimated diagnostic would seem to be good when the truth was bad (far from uniform). We can spot this by fitting a mixture of Beta CDFs as a diagnostic.

\subsubsection*{Consistency and reliability of $\hat D_y$ }
We showed consistency of our method in Section 3 under the assumption that the Beta-NN parameterisation is sufficiently expressive. The lemma holds without this assumption but convergence in probability holds for the parameter $w^*$ minimising the KL-divergence. The theorem holds if there exists $w\in \mathcal{R}^m$ such that $KL(F_y,G_y)>KL(F_y,D_y(G_y;w))$. If this is not the case then the approximation $G_y$ must be good! In this case the MLE converges to $w_I$, parameterising the identity map, a reasonable diagnostic outcome.

We cannot guarantee for any given $N$ that our estimate $\hat D_y$ based on $\hat w_N$ is reliable (it isn't, as it is only mapping closer to the truth in probability) so some diagnostics are needed to check our diagnostic tool. Section 3.1 lists two obvious validation checks on $\hat D_y$ and we may also vary the number of mixture components in the MDN. 

In principle we have access to an unlimited amount of data to learn $D_y$, if we can efficiently simulate the generative model. However, this type of check can be time consuming, as it requires repeated calls of the approximation scheme for each synthetic data point. This means our method if effective if the computational cost of the evaluating the approximation $G_y(x)$ is manageable. 

\newpage

\subsection{Proofs}

The following proposition reproduces a result given in \citet{papamakarios2016fast}.

\setcounter{theorem}{0}
\begin{proposition} Suppose the set $W$ in Equation~\ref{eq:W} is non-empty. Let $y_i \sim p(y), \ q_i \sim D_{y_{i}}(q)$ independently for $i=1,...,N$. Then $N^{-1}\ell(w,\{q_i,y_i\}_{i=1}^N)$ converges in probability to
\[-E_Y(\mbox{KL}(D_Y(\cdot),D_Y(\cdot;w)))+E_{Q,Y}(\log(d_Y(Q)).\] 
This limit function is maximized at $w\in W$.
\end{proposition}
\begin{proof}
Our presentation here is very brief as this result is known. We include this proof outline in order to make the meaning of the proposition clear.

By the WLLN,
\[N^{-1}\ell(w,\{q_i,y_i\}_{i=1}^N)\stackrel{P}{\rightarrow} E_{Q,Y}(\log(d_Y(Q;w))\] and the first statement follows as 
\begin{eqnarray*}
E_{Q,Y}(\log(d_Y(Q;w))&\!\!\!=\!\!\!&-E_Y(\mbox{KL}(D_Y(\cdot),D_Y(\cdot;w)))\\
&&+E_{Q,Y}(\log(d_Y(Q)).
\end{eqnarray*}
The second term does not depend on $w$ so we maximise the scaled limit of the log-likelihood by minimising the KL-divergence. Since $\mbox{KL}(D_y(\cdot),D_y(\cdot;w^*))=0$ for all $y\in \mathcal{Y}$ iff $D_y(q;w^*)=D_Y(q)$ at each $q,y$, and is otherwise continuous and positive, the limit function is maximised at $w^*\in W$ whenever this set is non-empty. 
\end{proof}

The result above shows that the {\it maximum of the limit} of the scaled log-likelihood gives the true distortion map. However a proof of consistency must show that the {\it limit of the maximum} of the scaled log-likelihood converges in probability to the set of parameter values that express the true distortion map. Standard theory for the MLE does not apply as the true parameter is not identifiable. The corresponding result for the non-identifiable case was given in \citet{redner1981note}. The following proof of consistency is based on that paper.

\setcounter{theorem}{0}
\begin{lemma}
Under the conditions of Proposition~1, the estimate $D_y(q;{\hat w}_N)$ is consistent, that is
\[
\lim_{N\rightarrow\infty}\Pr(|D_y(q;{\hat w}_N)-D_y(q)|>\epsilon)=0.
\]
for every fixed $q,y$.
\end{lemma}

\begin{proof}
Let $W=\{w^*: D_y(\cdot;w^*)=D_y(\cdot), \ y \in \mathcal{Y}\}$. Let $\tau(\mathcal{R}^m)$ be the quotient topological space defined by taking $\mathcal{R}^m$, choosing a point $W^*\in W$, and identifying all points in $W$ in the original space $\mathcal{R}^m$ with the single point $W^*$ in $\tau(\mathcal{R}^m)$.


We now show (by citing \cite{redner1981note}) that the maximum likelihood estimator converges in $\tau(\mathcal{R}^m)$ to $W^*$, $\hat w_N\stackrel{P}{\rightarrow}W^*$
as $N\rightarrow\infty$. This is not an immediate consequence of standard regularity conditions for the convergence of the MLE, as we do not assume that there is a unique $w^*$ satisfying $D_y(q)=D_y(q;w^*)$, so $w^*$ is not identifiable. In fact we can construct cases where $W$ has uncountably many elements, so this assumption does not hold.
However, the MLE convergence results for a non-identifiable parameter given in \cite{redner1981note} apply. Recall that $W^*$ is the point in the $\tau(\mathcal{R}^m)$ corresponding to the set $W$ in the original space $\mathcal{R}^m$. 
By Theorem 4 of \cite{redner1981note}, we have $\hat w_N \stackrel{a.s.}{\rightarrow} W^*$ as $N \rightarrow \infty$. All regularity conditions required for Theorem 4 of \cite{redner1981note} can be verified easily. 

It then follows from the continuity of $d_y(q;w)$ (and therefore $D_y(q;w)$) and the continuous mapping theorem that, for each pair $\{q,y\}$,
\[
D_y(q;\hat w_N)\stackrel{P}{\rightarrow} D_y(q;W^*)
\]
and then since $D_y(q;W^*)=D_y(q)$ we have
\[
D_y(\cdot;\hat w_N)\stackrel{P}{\rightarrow} D_y(\cdot).
\]

\end{proof}

\setcounter{theorem}{0}

\begin{theorem}
Under the conditions of Proposition~1 and assuming $KL(F_y,G_y)>0$, 
\[\Pr(KL(F_y,\hat F_y)<KL(F_y,G_y))\rightarrow 1\] as $N\rightarrow \infty$ for every fixed $y$.
\end{theorem}

\begin{proof}
$\hat F_y(x)=D_y({G}_y(x);\hat w_N)$ so the density of $\hat F_y$ is \[\hat \pi(x|y)=\tilde\pi(x|y)d_y({G}_y(x);\hat w_N).\] 
Recalling $\pi(x|y)=\tilde\pi(x|y)d_y({G}_y(x))$, we have
\begin{eqnarray*}
KL(F_y,\hat F_y)&\!\!\!\equiv\!\!\!&\int_{-\infty}^\infty \pi(x|y) \log\left(\frac{\pi(x|y)}{\hat \pi(x|y)}\right)dx\\
&\!\!\!=\!\!\!& \int_{-\infty}^\infty \!\!\! \tilde\pi(x|y)d_y({G}_y(x)) \log\left(\frac{d_y({G}_y(x))}{d_y({G}_y(x);\hat w_N)}\right)dx\\
&\!\!\!=\!\!\!& \int_{0}^1 d_y(q) \log\left(\frac{d_y(q)}{d_y(q;\hat w_N)}\right)dq\\
&\!\!\!=\!\!\!& KL(D_y(\cdot),D_y(\cdot;\hat w_N)),
\end{eqnarray*}
where we made the change of variables $q=G_y(x)$ to get from the second to third lines.
Taking $KL(F_y,G_y)=\epsilon$ with 
$\epsilon>0$ we have
\[
\Pr(KL(F_y,G_y)>KL(F_y,\hat F_y)) = \Pr(\epsilon>KL(D_y(\cdot),D_y(\cdot;\hat w_N))).
\]
By Lemma 1, $\displaystyle D_y(\cdot;\hat w_N)\stackrel{P}{\rightarrow} D_y(\cdot)$. The KL-divergence is a continuous mapping, so $KL(D_y(\cdot),D_y(\cdot;\hat w_N))\rightarrow 0$ in probability by the continuous mapping theorem. It follows that the limit as $N\rightarrow\infty$ of the quantity on the RHS of the last equality is equal one.
\end{proof}

Theorem 1 is a fairly natural consequence of Lemma 1: the procedure is Maximum-Likelihood, satisfies (some rather special) regularity conditions, and is therefore consistent. However we state the result in this form in order to emphasise that Algorithm 1 returns a distortion map that moves $\hat F_y$ closer to $F_y$, with high probability for all sufficiently large $N$, so that the map contains information about the distorting effects of the approximation, without actually sampling $F_y$, or even making it possible to sample $F_y$.

\subsection{Gene Fusion network}

We tried our approach on the larger Gene Fusion network \citep{hoglund2006gene, Kunegis:2013:KKN:2487788.2488173} with 291 nodes and 279 edges. Nodes represent genes and an edge is present if fusion of the two genes is observed during the emergence of cancer. The same ERGM given in Section 6 is used to fit the data.

In this example we report the ABC-reg and adj-lkd posteriors only, as the VI posterior behaves in the same way as in the Karate club network example (accurate mode, under-dispersed tails). Again, we report the fitted distortion map $\hat{D}$ and the recalibrated $\hat \pi(x^{(p)}|y_{obs})$ for each $p = 1,2,3$ for both approximation schemes in Fig.~\ref{fig:abcgene} and \ref{fig:adjgene}. 
Fig.~\ref{fig:abcgene} and \ref{fig:adjgene} show that estimated distortion maps $\hat{D}^{(p)}_{y_{obs}}$ are close to exact maps  $D^{(p)}_{y_{obs}}$ for each dimension for both approximation schemes. The estimated distortion map deviates from the identity map when the approximate marginals $\tilde \pi (x^{(p)}|y_{obs})$ deviate from the exact $\pi (x^{(p)}|y_{obs})$ substantially, and is close to the identity map when $\tilde \pi (x^{(p)}|y_{obs}) \approx \pi (x^{(p)}|y_{obs})$.

As in Section~6 we plot the distortion surfaces for the ABC-reg and adj-lkd posteriors for $\{x^{(1)}, x^{(3)}\}$. In this example there is little interesting bivariate structure as the joint distortion map is essentially the product of the univariate maps. From Fig.~\ref{fig:genesurface} we see the distortion surface of the ABC-reg posterior is far from 1, indicating that the ABC-reg approximation of the bivariate marginal posterior $\pi(x^{(1)},x^{(3)}|y_{obs})$ is unreliable. The distortion surface of adj-lkd posterior at the data is reasonably close to 1, though somewhat barrel-shaped, reflecting the fact that the approximation to $x^{(3)}$ is (fairly slightly) overdispersed. 

\begin{figure} \label{fig:gene}
    \centering
    \includegraphics[width=0.6\textwidth]{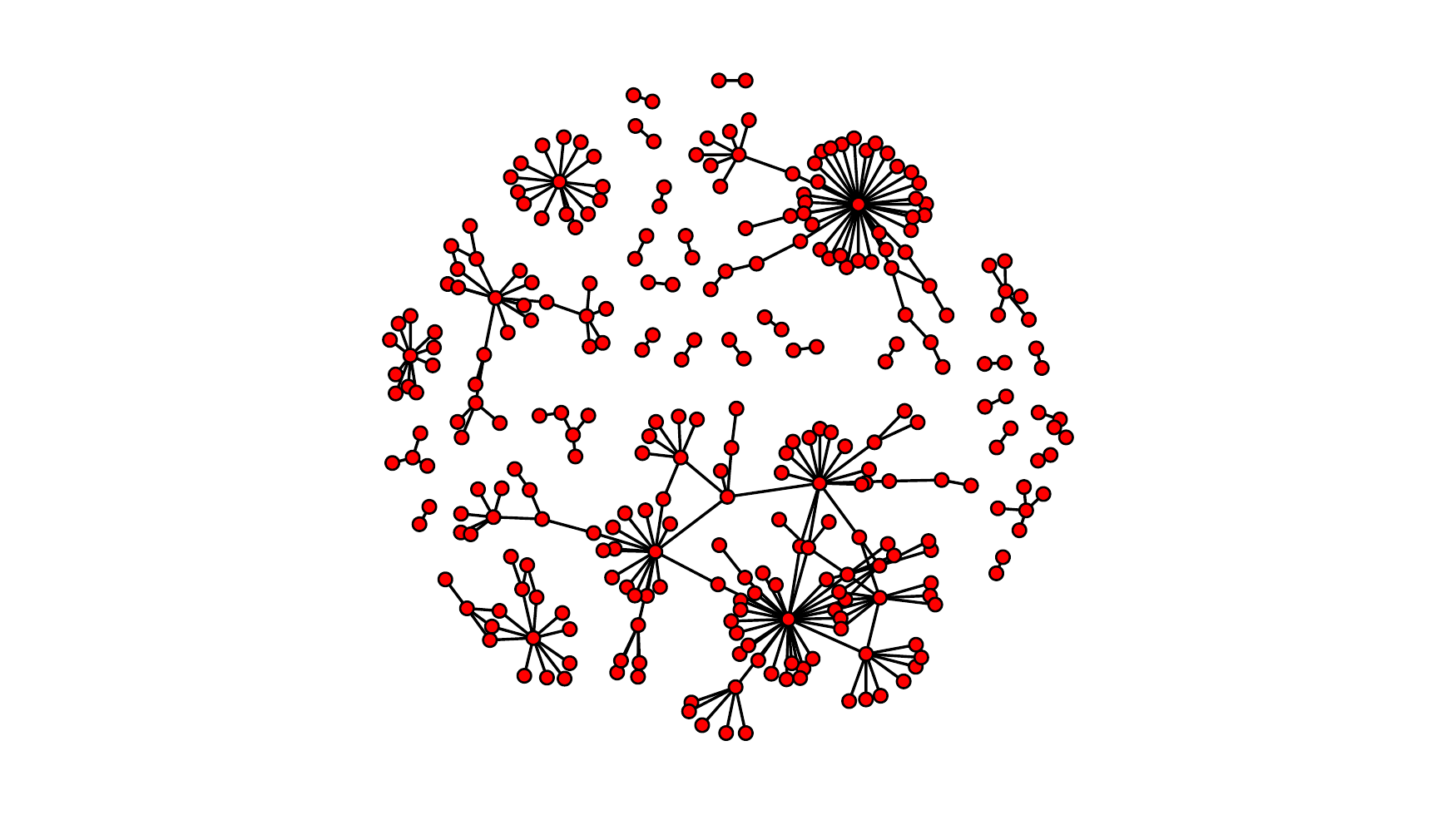}
    \caption{Gene fusion network}
    \label{fig:genefusion}
\end{figure}


\begin{figure}%
    \centering
    \subfloat{{\includegraphics[width=0.25\textwidth]{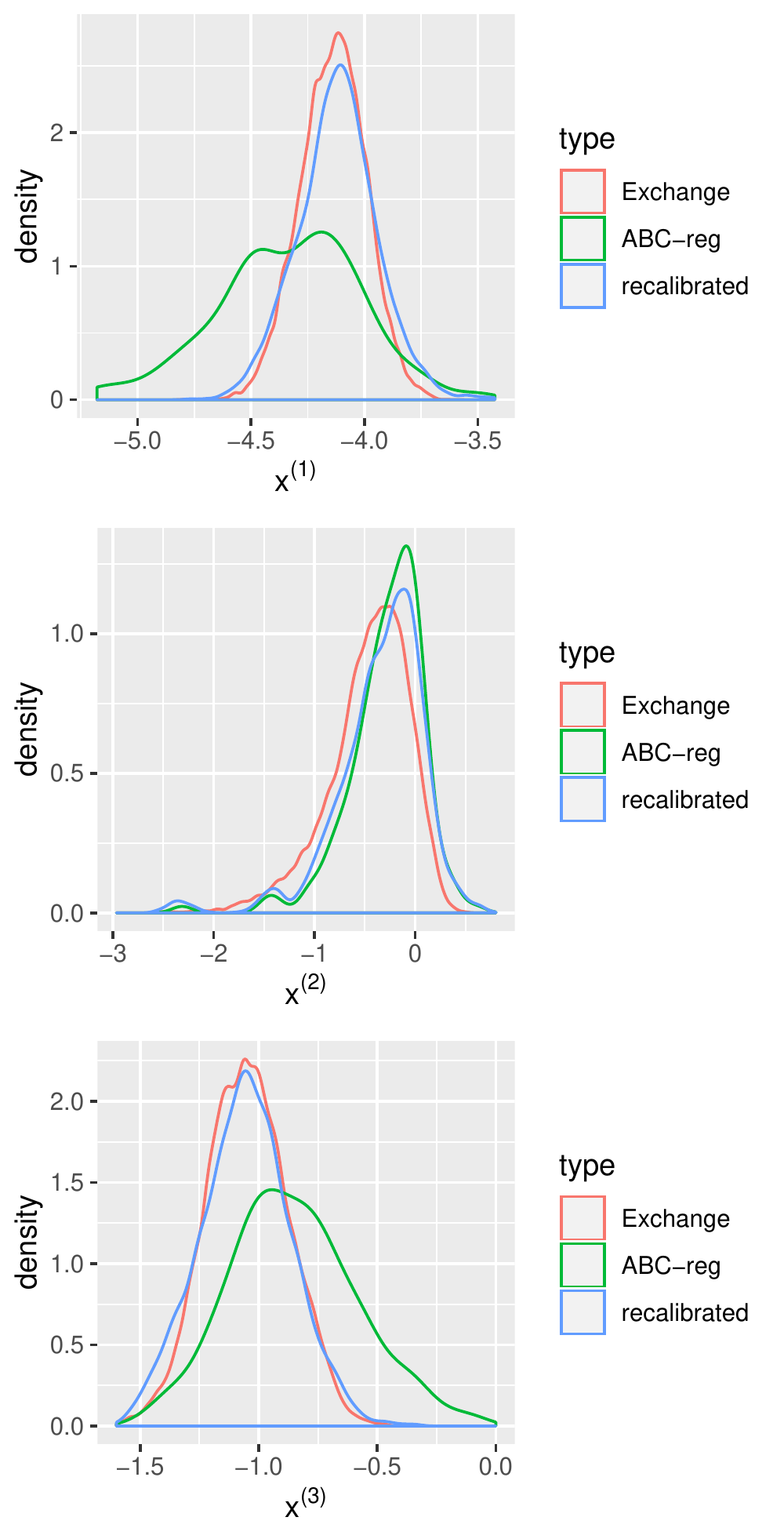} }}
    \subfloat{{\includegraphics[width=0.25\textwidth]{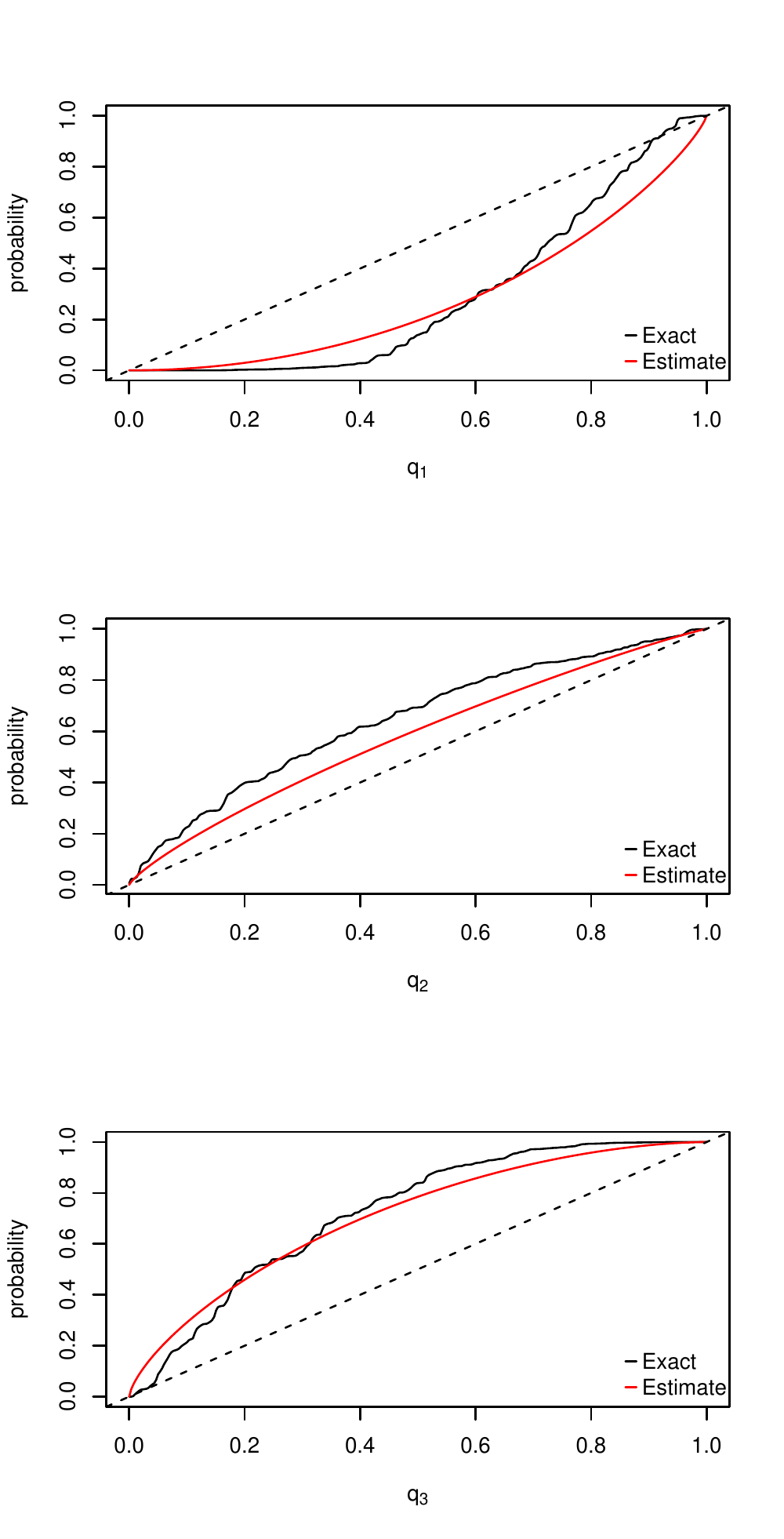} }}%
    \caption{Left: Recalibrated posterior of $x^{(p)}, \ p=1,...3$ for ABC-reg scheme Right: Exact $D^{(p)}_{y_{obs}}(\cdot)$ and fitted $\hat{D}_{y_{obs}}^{(p)}(\cdot)$ for $x^{(p)}$,  Dashed line represents the identity map.}%
    \label{fig:abcgene}%
\end{figure}

\begin{figure}%
    \centering
    \subfloat{{\includegraphics[width=0.25\textwidth]{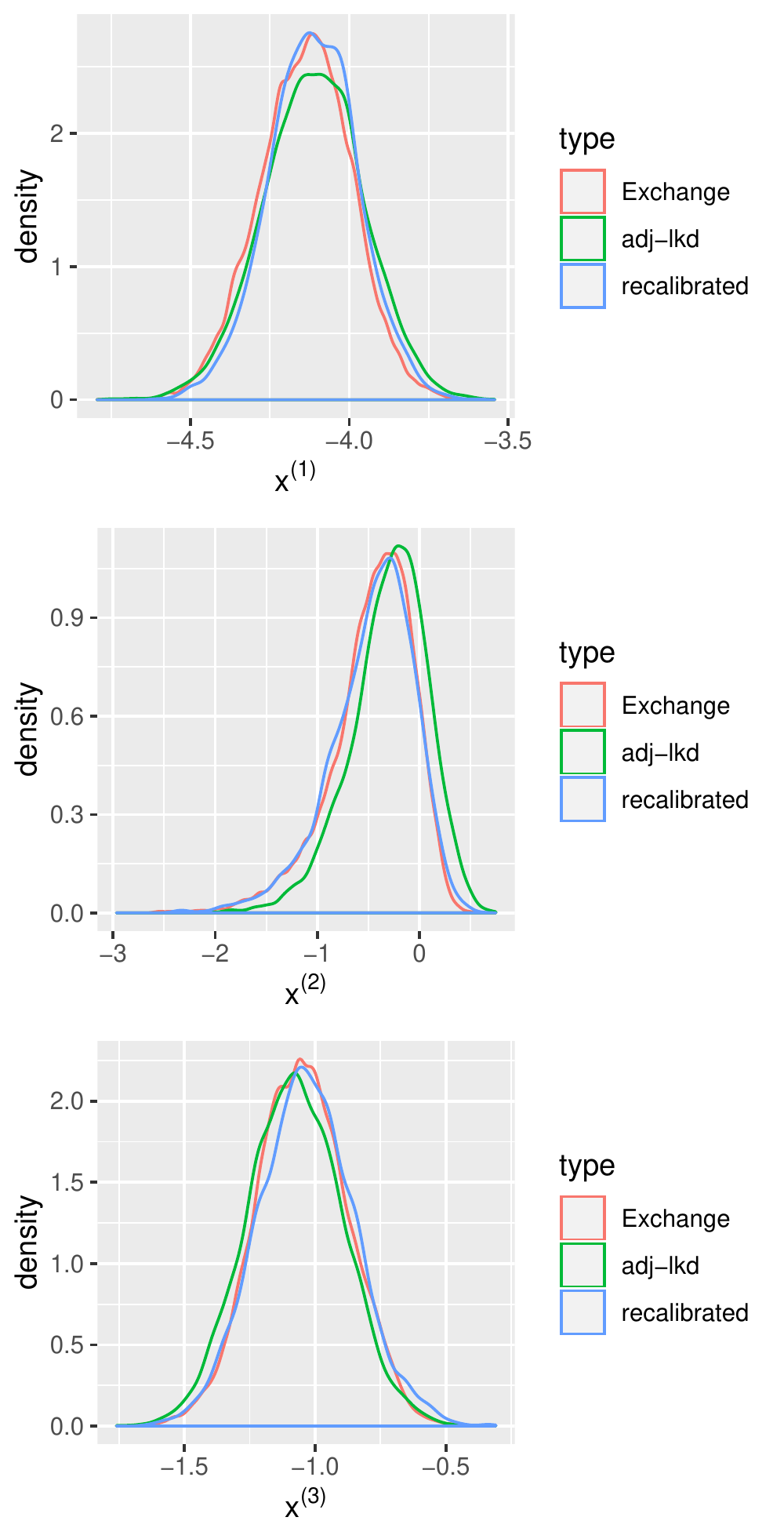} }}
    \subfloat{{\includegraphics[width=0.25\textwidth]{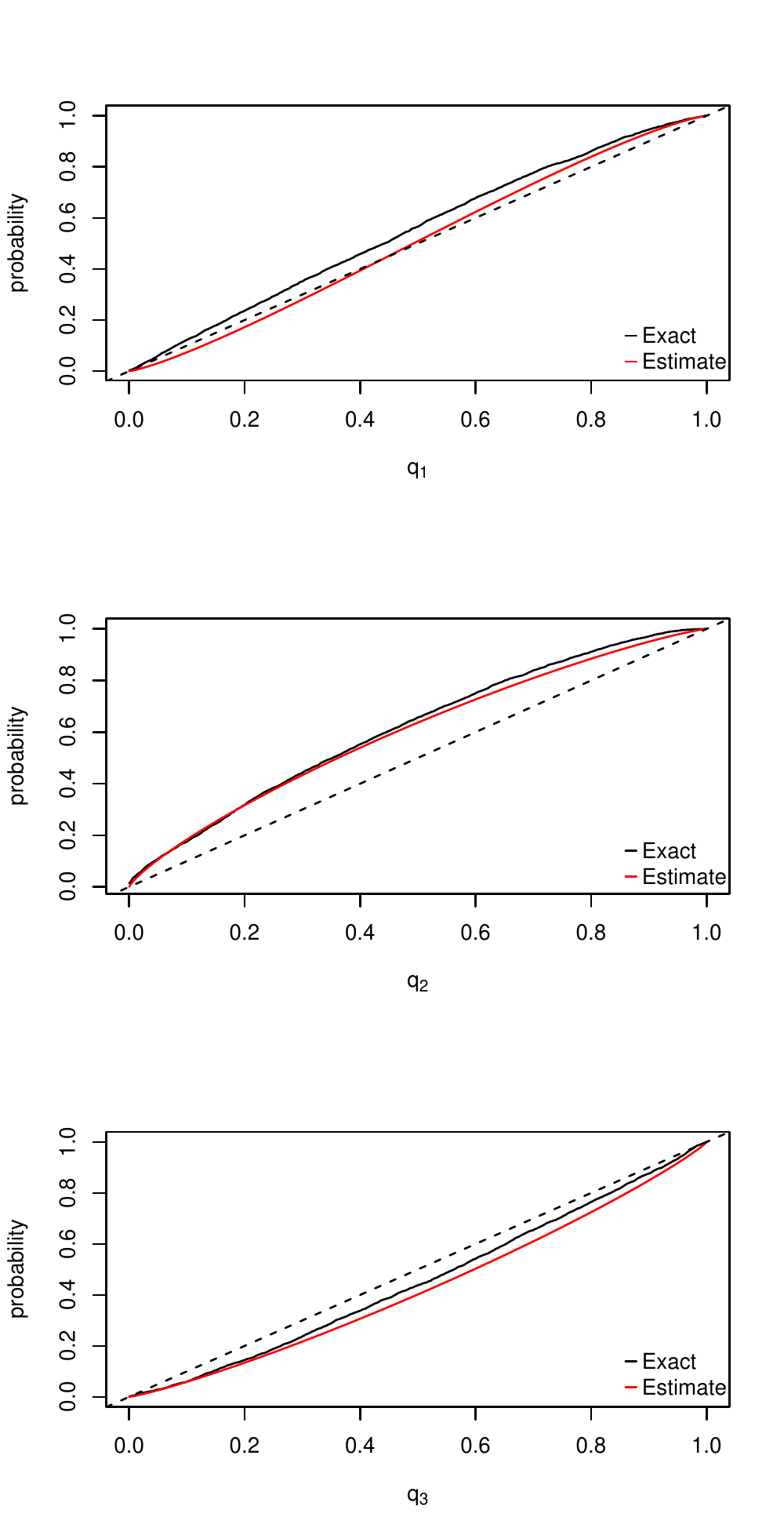} }}%
    \caption{Left: Recalibrated posterior of $x^{(p)}, \ p=1,...3$ for adj-lkd scheme Right: Exact $D^{(p)}_{y_{obs}}(\cdot)$ and fitted $\hat{D}_{y_{obs}}^{(p)}(\cdot)$ for $x^{(p)}$,  Dashed line represents the identity map.}%
    \label{fig:adjgene}%
\end{figure}

\begin{figure}%
    \centering
    \subfloat{{\includegraphics[width=0.27\textwidth]{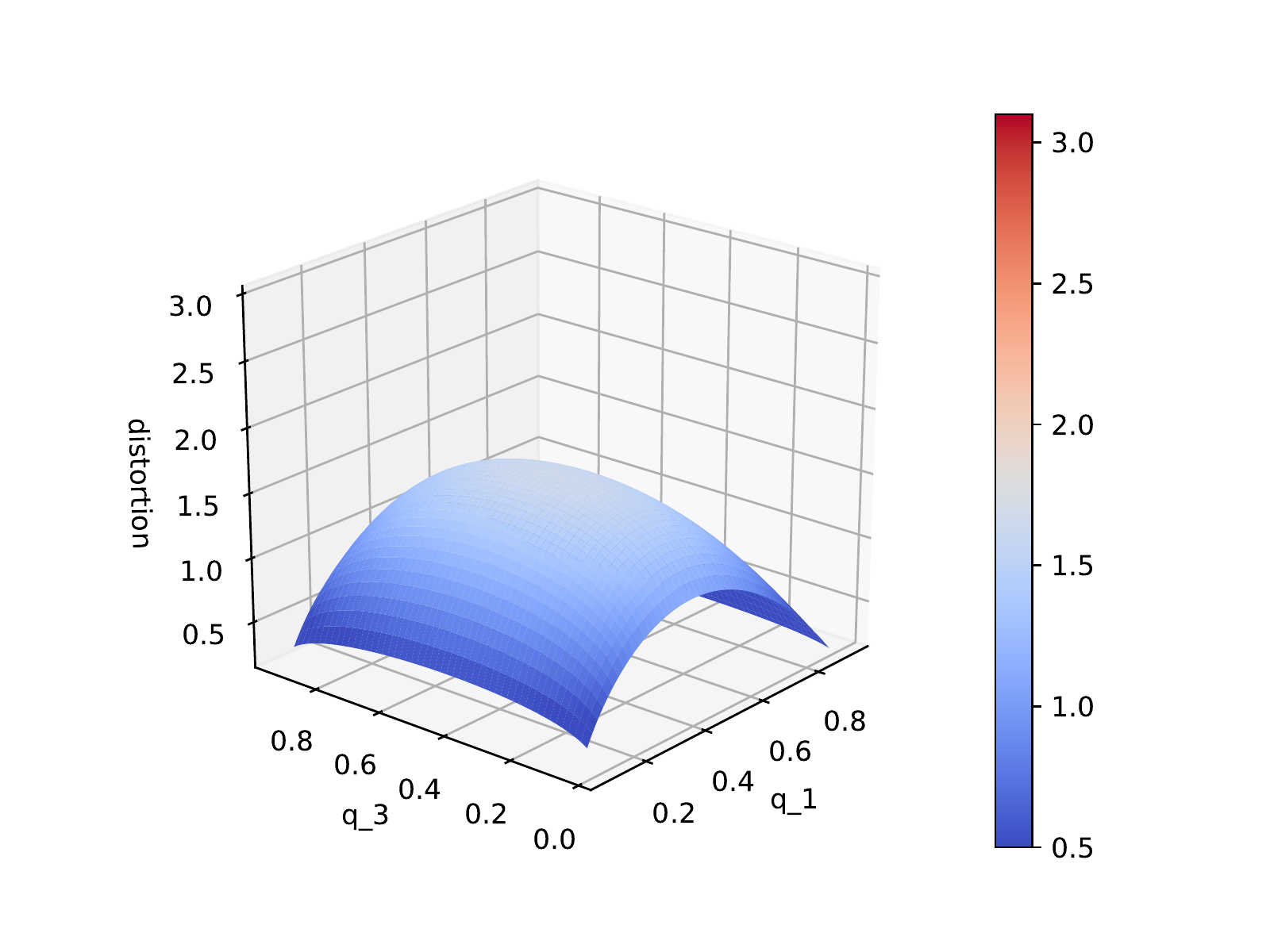} }}
    \subfloat{{\includegraphics[width=0.27\textwidth]{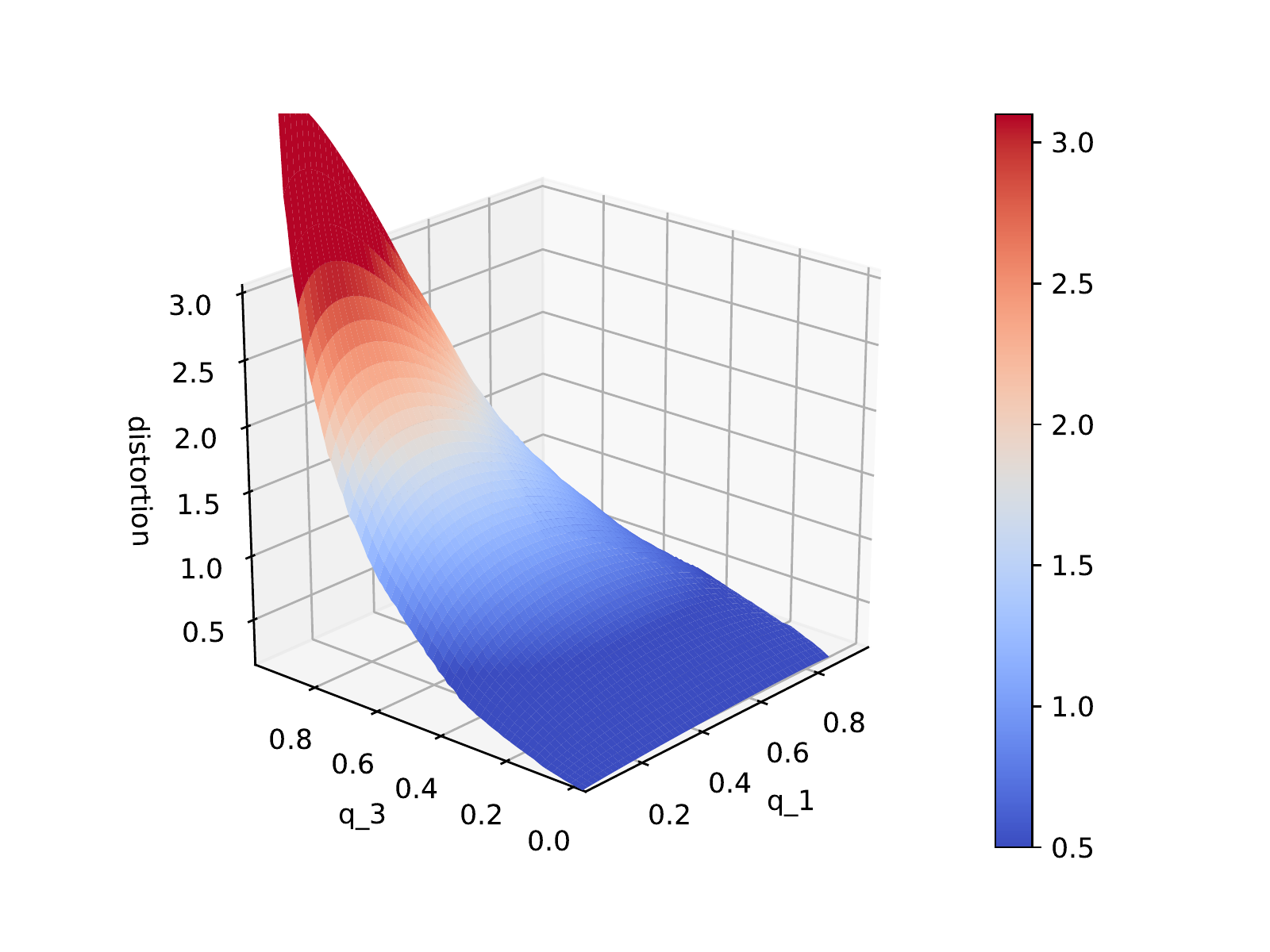} }}%
    \caption{Left: Distortion surface of adj-lkd posterior with respect to $q_1,q_3$. Right: Distortion surface of abc-reg posterior with respect to $x^{(1)},x^{(3)}$.}%
    \label{fig:genesurface}%
\end{figure}

\end{document}